\documentclass{article}

\usepackage[round]{natbib}

\providecommand{\tightlist}{%
  \setlength{\itemsep}{0pt}\setlength{\parskip}{0pt}}

\usepackage[a4paper]{geometry}
\usepackage[small,bf]{caption}

\usepackage[svgnames]{xcolor}
\usepackage[colorlinks,citecolor=MidnightBlue,linkcolor=FireBrick]{hyperref}

\usepackage{graphicx}
\usepackage{subcaption}
\usepackage{amsmath}
\usepackage{amssymb}
\usepackage{amsthm,thmtools}

\usepackage{mathtools}
\mathtoolsset{centercolon}

\usepackage{parskip}

\begingroup
    \makeatletter
    \@for\theoremstyle:=definition,remark,plain\do{%
        \expandafter\g@addto@macro\csname th@\theoremstyle\endcsname{%
            \addtolength\thm@preskip\parskip
            }%
        }
\endgroup
\usepackage{color}

\theoremstyle{plain}
\newtheorem{theorem}{Theorem}
\newtheorem{lemma}[theorem]{Lemma}

\newtheorem{corollary}{Corollary}

\theoremstyle{definition}
\newtheorem{definition}{Definition}
\theoremstyle{remark}
\declaretheorem[name=Remark, qed={\lower-0.3ex\hbox{$\triangleleft$}}]{remark}

\DeclareMathOperator{\sign}{sign}

\makeatletter
\newcommand{\pushright}[1]{\ifmeasuring@#1\else\omit\hfill$\displaystyle#1$\fi\ignorespaces}
\newcommand{\pushleft}[1]{\ifmeasuring@#1\else\omit$\displaystyle#1$\hfill\fi\ignorespaces}
\makeatother

\newcommand{\hideFromPandoc}[1]{#1}
\hideFromPandoc{

}

\usepackage{multicol}

\newcommand{\startsupp}{%
\newpage
\appendix %
}

\newcommand{\putacknowledgement}{}

\newcommand{\figureone}{
  \scantokens{
\begin{figure*}[t]
\centering
\begin{subfigure}[t]{.32\linewidth}
\centering\includegraphics[width=1.0\textwidth]{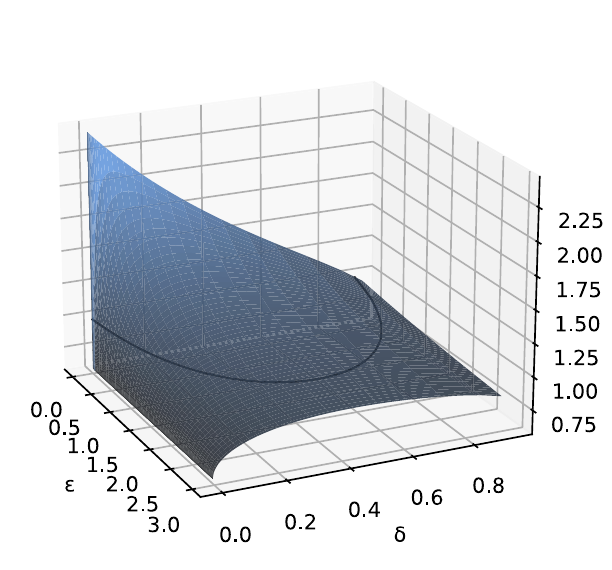}
\subcaption{$v_{1, \log}$}\label{fig:rat01}
\end{subfigure} 
\begin{subfigure}[t]{.32\linewidth}
\centering\includegraphics[width=1.0\textwidth]{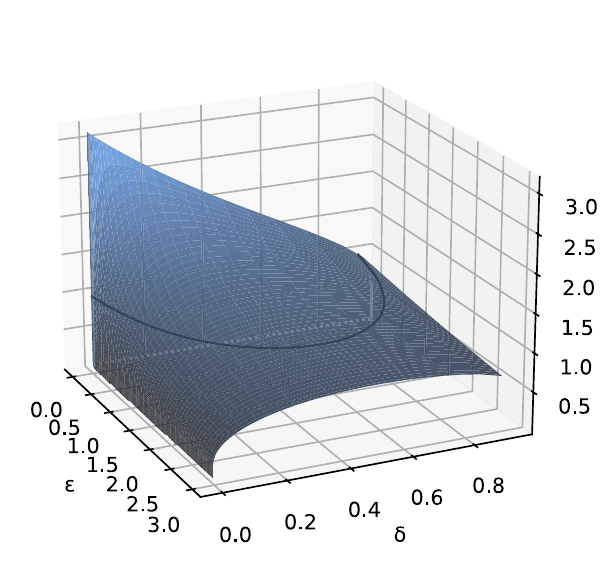}
\subcaption{$v_{1, 2}$}\label{fig:rat02}
\end{subfigure} 
\begin{subfigure}[t]{.32\linewidth}
\centering\includegraphics[width=1.0\textwidth]{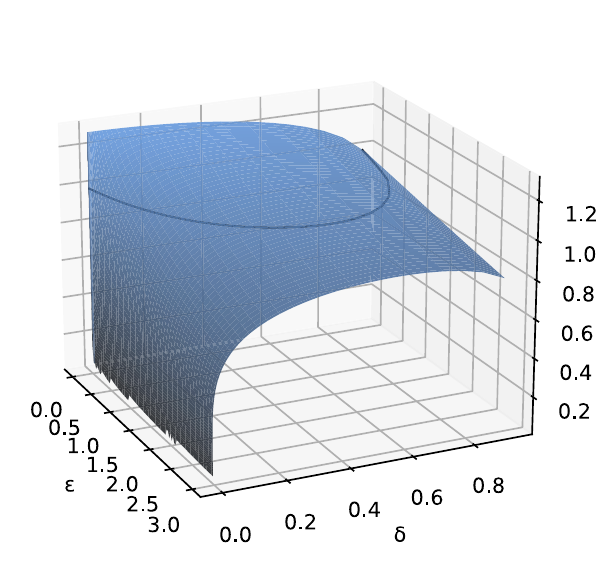}
\subcaption{$v_{\log, 2}$}\label{fig:rat12}
\end{subfigure}
\caption{\label{fig:rat} For pairs of Logistic($\log$), Laplace (1),
and Gaussian ({2}) distributions: plots of variance ratios
$v_{a,b}$ and contours where the resulting variances are equal.}
\end{figure*}}}

\newcommand{\figuretwo}{
  \scantokens{
\begin{figure*}[t]
\centering
\begin{subfigure}[t]{.24\linewidth}
\centering\includegraphics[width=\textwidth]{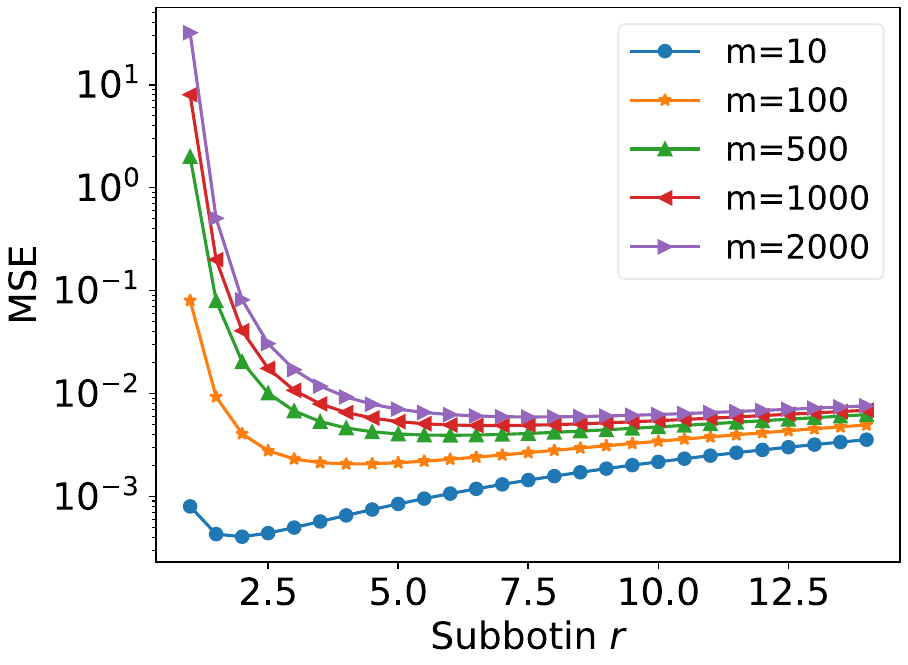}
\subcaption{Mean squared errors for $\epsilon=1$}\label{fig:vars}
\end{subfigure} 
\begin{subfigure}[t]{.24\linewidth}
\centering\includegraphics[width=\textwidth]{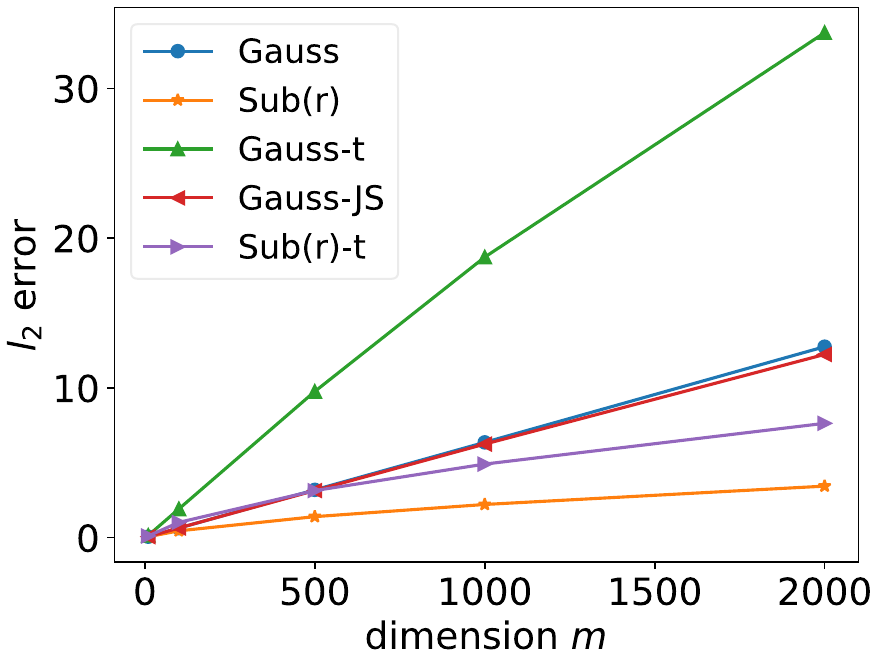}
\subcaption{$l_{2}$ errors for $\epsilon=1$}\label{fig:e1err}
\end{subfigure} 
\begin{subfigure}[t]{.24\linewidth}
\centering\includegraphics[width=\textwidth]{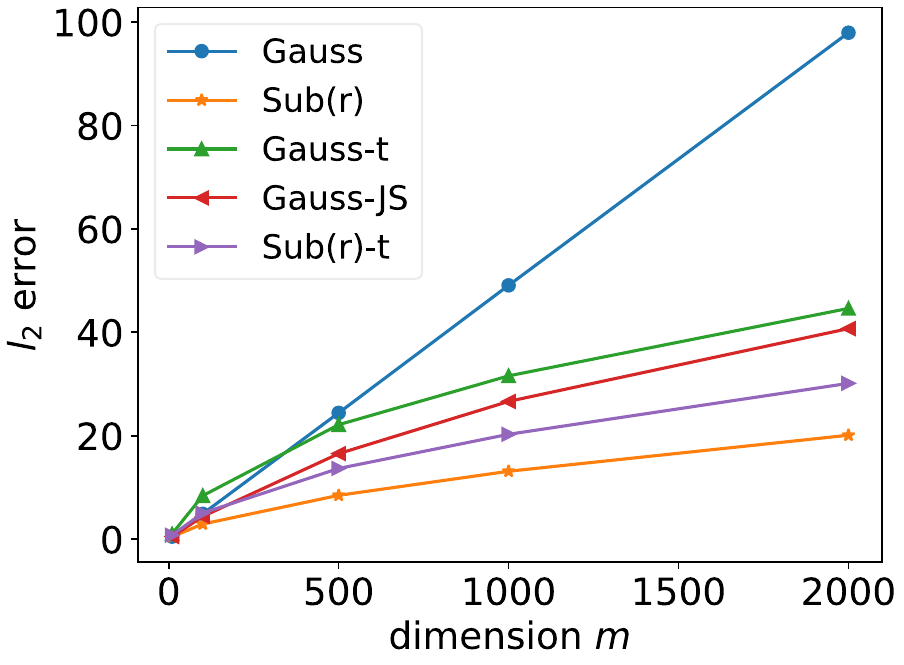}
\subcaption{$l_{2}$ errors for $\epsilon=0.1$}\label{fig:e01err}
\end{subfigure} 
\begin{subfigure}[t]{.24\linewidth}
\centering\includegraphics[width=\textwidth]{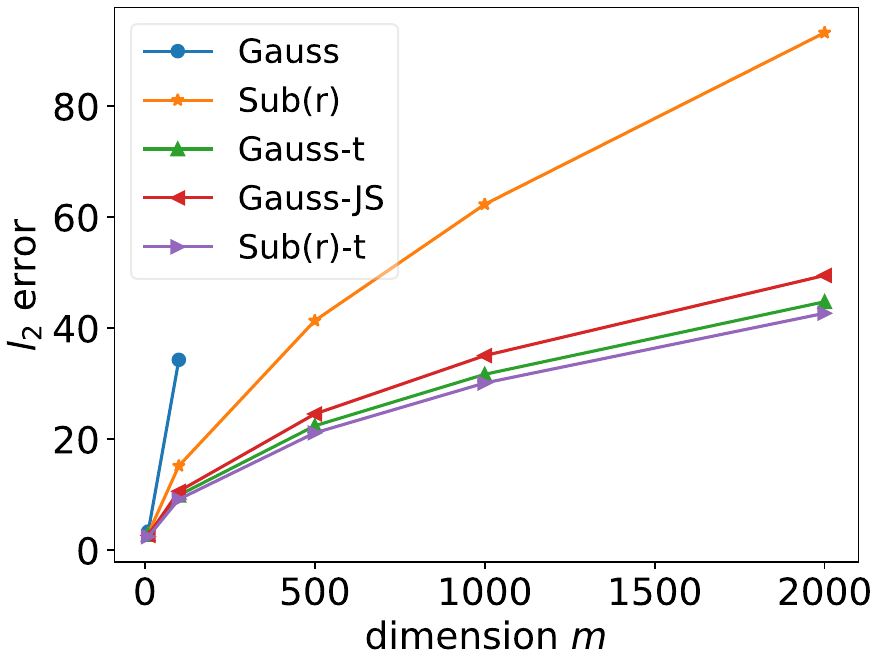}
\subcaption{$l_{2}$ errors for $\epsilon=0.01$}\label{fig:e001err}
\end{subfigure} 
\caption{\label{fig:experiment} Graphical summary of the mean estimation experiments.}
\end{figure*}}}

\newcommand{\figoneaistats}{}
\newcommand{\figtwoaistats}{}
\newcommand{\figoneother}{\figureone}
\newcommand{\figtwoother}{\figuretwo}

\usepackage{labelschanged}

\title{Differential privacy for symmetric log-concave
mechanisms\thanks{This work has been funded in part by Innlandet
Fylkeskommune, as well as Research Council of Norway grants 308904 and
309834. We thank the anonymous reviewers for helpful comments.}}

\author{Staal A.
Vinterbo\footnote{Department of Information Security and Communication Technology, Norwegian University of Science and Technology.}}
\date{}

\usepackage{tikz}
\usetikzlibrary{tikzmark}

\newcounter{changebar}
\newcommand{\changebarcolor}{red}
\newcommand{\changestart}[1][red]{%
    \renewcommand{\changebarcolor}{#1}%
    \stepcounter{changebar}%
    \tikzmarknode{chbar-\thechangebar-start}{\strut}%
}
\newcommand{\changeend}{%
    \tikzmarknode{chbar-\thechangebar-end}{\strut}%
    \begin{tikzpicture}[remember picture, overlay]
        \draw[very thick, \changebarcolor] ([xshift={\oddsidemargin+1in-10pt}]current page.west |- chbar-\thechangebar-start.north) -- ([xshift={\oddsidemargin+1in-10pt}]current page.west |- chbar-\thechangebar-end.south);
    \end{tikzpicture}%
}

\begin{document}

\maketitle

\begin{abstract}
Adding random noise to database query results is an important tool for
achieving privacy. A challenge is to minimize this noise while still
meeting privacy requirements. Recently, a sufficient and necessary
condition for \((\epsilon, \delta)\)-differential privacy for Gaussian
noise was published. This condition allows the computation of the
minimum privacy-preserving scale for this distribution. We extend this
work and provide a sufficient and necessary condition for
\((\epsilon, \delta)\)-differential privacy for all symmetric and
log-concave noise densities. Our results allow fine-grained tailoring of
the noise distribution to the dimensionality of the query result. We
demonstrate that this can yield significantly lower mean squared errors
than those incurred by the currently used Laplace and Gaussian
mechanisms for the same \(\epsilon\) and \(\delta\).
\end{abstract}

\newcommand{\xqed}[1]{%
  \leavevmode\unskip\penalty9999 \hbox{}\nobreak\hfill
  \quad\hbox{\ensuremath{#1}}}

\newcommand{\eprop}{}

\newcommand{\mc}[1]{\ensuremath{\mathcal{#1}}}
\newcommand{\mb}[1]{\ensuremath{\mathbf{#1}}}
\newcommand{\mbb}[1]{\ensuremath{\mathbb{#1}}}

\newcommand{\xcirc}[1]{\vcenter{\hbox{$#1\circ$}}}
\newcommand{\fcomp}{\mathbin{\mathchoice
  {\xcirc\scriptstyle}
  {\xcirc\scriptstyle}
  {\xcirc\scriptscriptstyle}
  {\xcirc\scriptscriptstyle}
}}

\newcommand{\R}{\ensuremath{\mathbb{R}}}
\newcommand{\N}{\ensuremath{\mathbb{N}}}

\newcommand{\X}{\ensuremath{\mathcal{X}}}
\newcommand{\Y}{\ensuremath{\mathcal{Y}}}
\newcommand{\cov}{\operatorname{cov}}
\newcommand{\var}{\operatorname{Var}}
\newcommand{\E}{\operatorname{E}}
\newcommand{\tr}{\operatorname{Tr}}
\newcommand{\strings}{\ensuremath{\mathbb{S}}}

\newcommand{\mult}{\operatorname{Mult}}
\newcommand{\Bet}{\ensuremath{\mathcal{B}}}
\newcommand{\Lap}{\ensuremath{\mathcal{L}}}
\newcommand{\gam}{\operatorname{Gamma}}
\newcommand{\Exp}{\operatorname{Exp}}
\newcommand{\bit}{\mbb{B}}
\newcommand{\biti}{\ensuremath{\bit^\infty}}
\newcommand{\ms}[1]{\ensuremath{\mathit{#1}}}
\newcommand{\deq}{\ensuremath{\stackrel{\tiny d}{=}}}
\newcommand{\dto}{\ensuremath{\stackrel{\tiny d}{\rightarrow}}}
\newcommand{\rto}{\ensuremath{\stackrel{\tiny\mc{R}}{\rightarrow}}}

\newcommand{\fault}{\ensuremath{\mathrm{fault}}}

\newcommand{\Sketch}{\textsc{Sketch}}
\newcommand{\Assume}{\textsc{Assume}}
\newcommand{\Prove}{\textsc{Prove}}
\newcommand{\Case}{\textsc{Case}}
\newcommand{\Proof}{\textsc{Proof}}
\newcommand{\Let}{\textsc{Let}}

\newcommand{\M}{\mathcal{M}}
\newcommand{\share}{\operatorname{share}}
\newcommand{\csum}{\operatorname{sum}}

\newcommand{\proto}{\textsc{MPLM}}
\newcommand{\corr}{\operatorname{corr}}
\newcommand{\1}{\ensuremath{\mathbf{1}}}
\newcommand{\0}{\ensuremath{\mathbf{0}}}

\newcommand{\diag}{\operatorname{diag}}

\newcommand{\Ri}{\ensuremath{\mbb{R}_{I}}}

\newcommand{\fmod}{\operatorname{fmod}}
\newcommand{\imod}{\operatorname{imod}}

\newcommand{\erf}{\operatorname{erf}}

\newcommand{\mats}{\ensuremath{\mathcal{M}}}

\newcommand{\lfun}{\ensuremath{\lambda}}

\newcommand{\Sub}[1]{\ensuremath{\text{Subbotin}_{#1}}}
\newcommand{\MSE}{\operatorname{MSE}}

\newcommand{\vvv}[1]{\ensuremath{\mathbf{#1}}}
\newcommand{\w}[1]{\ensuremath{\mathbb{#1}}}

\newcommand{\supp}{\operatorname{Supp}}

\hypertarget{introduction}{%
\section{INTRODUCTION}\label{introduction}}

\changestart \textbf{Update:} Lemma \ref{lem:iso} is invalid, as is
Theorem \ref{thm:gensub} building on it. This also invalidates
essentially all the results of of Section \ref{sec:extensions} and the
corresponding conclusions. Thanks to the community for pointing out the error in Lemma \ref{lem:iso}. \changeend

The Hippocratic oath, estimated to be about 2500 years old, mentions
patient privacy. Privacy is also recognized as a fundamental human right
\citep{UNUDHR} and is featured in some form in the constitutions of more
than 180 countries \citep{elkinsConstituteWorldsConstitutions2013}.
Still, privacy, and data privacy in particular, is very much an unsolved
problem.

Data privacy rests on keeping the data itself confidential as well as
controlling information leakage by processes that have access to the
data. In the following we consider leakage from database query results.
A problem is that the result of a database query can be highly revealing
about an individual. For example, consider the case where all records
except one are known to an adversary. Now, any result of any linear
query immediately reveals the query value of the unknown record.

The definition of \((\epsilon, \delta)\)-differential privacy
\citep{Dwork2006a, our-data-ourselves-privacy-via-distributed-noise-generation}
addresses the problem of revealing query responses by requiring query
response algorithms to be insensitive to single record changes. This
insensitivity is achieved by random perturbations in the computation of
query results and is quantified by that any single record change is not
allowed to change the likelihood of any output event by more than a
multiplicative factor \(e^{\epsilon}\) plus a small additive term
\(\delta\). For small \(\epsilon\) and very small \(\delta\), this
provides a strong form of privacy but still allows producing high
quality statistics, particularly when they are informative and the data
is representative \citep{10.1145/1536414.1536466}.

For real-valued queries, the prototypical differentially private
response algorithms add noise distributed according to either the
Laplace or Gaussian distributions to the real query answer. These
algorithms are called the Laplace and Gaussian mechanisms, respectively
(see e.g., \citet{dworkAlgorithmicFoundationsDifferential2014}).

For these mechanisms that add symmetric noise to query results,
differential privacy places a condition on the scale parameter of the
noise distributions. More privacy requires larger scale. Of course, a
larger scale decreases the concentration of the now noisy query result
around its true value. Therefore, a goal is to find the minimal scale
that meets the requirement of differential privacy.

Pursuing this goal, there exists a line of work into sufficent
conditions on the scale for the \((\epsilon, \delta)\)-differentially
private Gaussian mechanism, starting with what could be called the
standard closed form scale bound for the Gaussian mechanism noted in
\citet{Dwork2006a} (and also discussed in
\citet{dworkAlgorithmicFoundationsDifferential2014}), via an improvement
in the underlying analytic sufficient condition by \citet{6606817}, to
the analytic sufficient and necessary condition by
\citet{pmlr-v80-balle18a}. To our knowledge, the Gaussian distribution
is the only continuous distribution for which a necessary and sufficient
condition for \((\epsilon, \delta)\)-differential privacy has been
described so far.

Inspired by the above, we extend the line of work by a sufficient and
necessary condition for \((\epsilon, \delta)\)-differential privacy for
mechanisms that add noise distributed according to log-concave and
symmetric densities.

Our results allow us to demonstrate the optimization of utility by not
only minimizing the scale of a fixed mechanism's perturbations to meet
the privacy requirements, but also simultaneously choose the mechanism
itself by taking the dimensionality of the query function into account.

More background and other related works are described in Section
\ref{sec:background}.

\hypertarget{summary-of-findings}{%
\subsection{Summary of findings}\label{summary-of-findings}}

Let \(d\) be a database, \(q\) be a real-valued query function on
databases, \(s \geq 0\), and \(X\) be a random variable distributed
according to density \(f\). We will call an algorithm that returns a
variate of the random variable \(q(d) + sX\) a ``mechanism'' and often
identify the mechanism by its random variable. When \(f = e^{-\psi}\)
where \(\psi\) is even and convex, we call the associated mechanism
symmetric log-concave.

Our main theoretical findings are

\begin{itemize}
\item
  Lemma \ref{lem:neccsuff} that states a sufficient and necessary
  condition for \((\epsilon, \delta)\)-differential privacy for
  symmetric log-concave mechanisms for real valued query functions, and
\item
  Lemma \ref{lem:iso} that states that the condition in Lemma
  \ref{lem:neccsuff} can be extended naturally to query functions with
  values in \(\ensuremath{\mathbb{R}}^{n}\) for \(n\geq 1\) and
  mechanisms that add noise vectors distributed according to spherically
  symmetric log-concave distributions that are spherical with respect to
  the norm used to define the global sensitivity of the query. In
  particular, this is the case for noise vectors of iid
  \ensuremath{\text{Subbotin}_{p}} random variables and \(p\)-norms for
  \(p\geq 1\) as described in Theorem \ref{thm:gensub}, giving rise to
  an infinite family of mechanisms containing both the Laplace and the
  Gaussian mechanisms.
\end{itemize}

These results enable maximizing utility by not only finding the optimal
member of a fixed noise scale family such as the Gaussian, but also
simultaneously selecting the scale family itself depending on the query
response dimensionality. In our case, given the one-to-one
correspondence between vector valued \ensuremath{\text{Subbotin}_{p}}
mechanisms and the \(p\)-norm used to define the global sensitivity
\(\Delta\) of the query function, we select the scale family by
including the parameter \(p\) into the parameter optimization process.
We demonstrate this idea by empirically showing that depending on the
number of columns in a real valued data table, there are different
\ensuremath{\text{Subbotin}_{p}} mechanisms that yield the smallest
\(l_{2}\)-error for estimating the average row under given privacy
constraints. As the number of columns increases, the error of the
selected mechanism becomes much smaller than if just considering the
Laplace or Gaussian scale families.

In addition, we show that a minimum scale bound that fulfills the
condition in Lemma \ref{lem:neccsuff} is linear in \(\Delta\) (Lemma
\ref{lem:standards}), we produce closed form necessary and sufficient
conditions for \((\epsilon, \delta)\)-differential privacy and scale
\(s\) \begin{align*}
  s &\geq \frac{\Delta}{\epsilon - 2 \log\left(1-\delta \right)},
      \text{and}\\
  s &\geq \frac{\Delta}{2 \log\left(\frac{e^{\frac{\epsilon}{2}}+
     \sqrt{\delta \left(e^{\epsilon}+\delta  
     -1\right)}}{1-\delta}\right)},
\end{align*} for the Laplace and Logistic mechanisms, respectively
(Theorems \ref{thm:laplace} and \ref{thm:logistic}). In the case of
mechanisms supported on \(\ensuremath{\mathbb{R}}\), we prove that
one-dimensional mechanisms that can achieve
\((\epsilon,0)\)-differential privacy can exhibit arbitrarily smaller
variances than those that that cannot for small enough \(\delta\)
(Theorem \ref{thm:unboundeds}). We demonstrate that this \(\delta\) need
not be very small by numerically showing that the Laplace and Logistic
mechanisms, which can achieve \((\epsilon, 0)\)-differential privacy,
exhibit smaller variance than the Gaussian, that cannot, for a
significant range of privacy parameters, for example when
\(\epsilon \geq 0.05\) and \(\delta \leq 0.001\).

\hypertarget{sec:prelims}{%
\section{PRELIMINARIES}\label{sec:prelims}}

We briefly recapitulate select definitions and known results.

\hypertarget{subbotin-densities}{%
\subsection{Subbotin densities}\label{subbotin-densities}}

Let \(C(r) := 2 \Gamma \left(\frac{1}{r}\right) r^{\frac{1}{r}-1}\)
where \(\Gamma\) denotes the gamma function, and let \begin{align*}
f_{r}(x) := \frac{e^{-\frac{|x|^{r}}{r}}}{C(r)}.
\end{align*} The family of densities \(f_{r}\) for \(r>0\) are called
Subbotin \citep{subbotinLawFrequencyError1923} densities, exponential
power distribution densities, or the generalized normal distribution
\citep{doi:10.1080/02664760500079464} densities. This family includes
the standard Laplace density as \(f_{1}\), the standard Gaussian density
as \(f_{2}\), and the uniform density on \((-1,1)\) in the limit as
\(r\to \infty\). Let \(\Gamma(s, x)\) denote the upper incomplete gamma
function. Then, \begin{align*}
F_{r}(x) &:= \int_{-\infty}^{x} f_{r}(w) dw \\
&= \frac{1}{2}+\sign\left(x \right)\left(\frac{1}{2}-\frac{\Gamma
\left(\frac{1}{r}, \frac{{\mid x \mid}^{r}}{r}\right)}{2 \Gamma
\left(\frac{1}{r}\right)}\right),\\
F^{-1}_{r}(p) &= \sign\left(p-\frac{1}{2}\right)
F^{-1}_{\Gamma}\left(2|p-\frac{1}{2}|, \frac{1}{r}, \frac{1}{r}\right)^{\frac{1}{r}},
\end{align*} where \(F^{-1}_{\Gamma}(p, \alpha, \beta)\) is the inverse
CDF of the Gamma distribution with density
\(f_{\Gamma}(x)=\frac{\beta^{\alpha}x^{\alpha-1}e^{-\beta x}}{\Gamma(\alpha)}\).
For a random variable \(X\) distributed according to \(f_{r}\),
\begin{align*}
\operatorname{Var}(X) &= r^{\frac{2}{r}} \frac{\Gamma \left( \frac{3}{r} \right)}
{\Gamma \left( \frac{1}{r} \right)}.
\end{align*} We will denote the Subbotin distribution with density
\(f_{r}\) as \ensuremath{\text{Subbotin}_{r}}.

The function \begin{align*}
\psi_{r}(x) := \frac{|x|^{r}}{r} + \log(C(r)) = -\log(f_{r}(x))
\end{align*} is even as \(x\) only enters as \(|x|\), and also convex
for \(r\geq 1\) as \(\psi_{r}''(x) = \frac{|x|^{r}(r-1)}{x^{2}} \geq 0\)
then. Consequently, the Subbotin densities are even and log-concave for
\(r \geq 1\).

\hypertarget{differential-privacy}{%
\subsection{Differential privacy}\label{differential-privacy}}

Formally, let a database \(d\) be a collection of record values from
some set \(V\). Two databases \(d\) and \(d'\) are \emph{neighboring} if
they differ in one record. Let \(\ensuremath{\mathcal{N}}\) be the set
of all pairs of neighboring databases.

\begin{definition}

The global sensitivity of a real-valued function \(q\) on databases is
\[
\Delta := \sup_{(d,d') \in \ensuremath{\mathcal{N}}} |q(d) - q(d')|.
\]

\end{definition}

\begin{definition}[$(\epsilon,\delta)$-differential privacy \cite{Dwork2006a,our-data-ourselves-privacy-via-distributed-noise-generation}]

A randomized algorithm \(M\) is called
\((\epsilon,\delta)\)-differentially private if for any measurable set
\(S\) of possible outputs and all
\((d,d') \in \ensuremath{\mathcal{N}}\) \[
\Pr(M(d) \in S) \leq e^{\epsilon}\Pr(M(d') \in S) + \delta,
\] where the probabilities are over randomness used in \(M\). By
\(\epsilon\)-differential privacy we mean \((\epsilon, 0)\)-differential
privacy.

\end{definition}

We are now ready to develop our results.

\hypertarget{the-one-dimensional-case}{%
\section{THE ONE-DIMENSIONAL CASE}\label{the-one-dimensional-case}}

\hypertarget{sec:contrib}{%
\subsection{A condition for real-valued mechanisms}\label{sec:contrib}}

The following is our first main result.

\begin{lemma}\label{lem:neccsuff}

Let \(X\) be a random variable distributed according to density
\(f(x) = e^{-\psi(x)}\) where
\(\psi : \ensuremath{\mathbb{R}}\to (-\infty,\infty]\) is convex and
even, with \(f\) having support \((-a, a)\) for \(a \in (0,\infty]\).
Then for a real-valued function \(q\) on databases with global
sensitivity \(\Delta \geq 0\) and a database \(d\), the mechanism
returning a variate of \(q(d) + s X\) is
\((\epsilon, \delta)\)-differentially private for \(\epsilon \geq 0\)
and \(\delta \geq 0\) if and only if for
\(F(x) = \int_{-\infty}^{x} f(u)du\) \begin{align} \label{eq:neccsuff} 
F \left( \frac{\Delta - t}{s} \right) - 
e^{\epsilon}F \left(  -\frac{t}{s} \right)
\leq \delta 
\end{align} where \begin{align} 
\label{eq:tstar} 
t := \sup \left\{z < as \mid \psi\left(\frac{z}{s}\right) -
\psi\left(\frac{z - \Delta}{s} \right) \leq \epsilon\right\}.
\end{align} Furthermore, if the above holds for scale \(s > 0\), then it
also holds for scales \(s' > s\).

\end{lemma}

\begin{remark}\label{rem:epsdiff}

Whenever \(t = as\), the criterion (\ref{eq:neccsuff}) reduces to
\(F \left(\frac{\Delta}{s} - a\right) \leq \delta\). When in addition
\(a = \infty\) and therefore \(t=\infty\), then the mechanism is
\((\epsilon,\delta)\)-differentially private for all \(\delta \geq 0\).

When \(a < \infty\) and \(\Delta \geq 2as\), we get that
\(f((x - \Delta)/s)/f(x/s) = 0\) for all \(x \in (-as,as)\) and
therefore the mechanism is not \((\epsilon, \delta)\)-differentially
private for any \(\delta < 1\). \qedhere

\end{remark}

The smallest finite scale \(s\) for which (\ref{eq:neccsuff}) holds is
monotone in \(\epsilon\), \(\delta\), and \(\Delta\). We make this
precise in the following.

\begin{lemma}\label{lem:standards}

Let \(s(\epsilon, \delta, \Delta)\) be the smallest \(s\) such that
(\ref{eq:neccsuff}) holds for log-concave and even density \(f\) and
real valued function \(q\) on databases with global sensitivity
\(\Delta > 0\). Then
\(s(\epsilon, \delta, \Delta) = \Delta\, s(\epsilon, \delta, 1)\), and
\(s(\epsilon, \delta, \Delta)\) is non-increasing in both \(\epsilon\)
and \(\delta\).

\end{lemma}

Unless otherwise evident, let in the following
\(s(\epsilon, \delta, \Delta)\) be defined as Lemma \ref{lem:standards}.

\hypertarget{sec:selected}{%
\subsection{Selected analytic conditions}\label{sec:selected}}

The following closed form condition is a tightening of Proposition 1 in
\citet{Dwork2006a}.

\begin{theorem}\label{thm:laplace}

Let \(X\) be a standard Laplace variable, and let \(q\) be a real valued
function on databases with global sensitivity \(\Delta\). Then the
mechanism that returns a variate of \(q(d) + s X\) is
\((\epsilon,\delta)\)-differentially private if and only if
\(s \geq s_{1}(\epsilon, \delta, \Delta) := \frac{\Delta}{\epsilon - 2 \log\left(1-\delta \right)}\).

\end{theorem}

\begin{remark}\label{rem:trlap}

An example of a mechanism that adds symmetric truncated log-concave
noise is the truncated Laplacian mechanism \citep{pmlr-v108-geng20a}
which Geng et al.~use to establish upper and lower error bounds for
\((\epsilon, \delta)\)-differential mechanisms. Given a standard Laplace
variable \(Y\), this mechanism is constructed by truncating the support
of \((\frac{\Delta}{\epsilon}) Y\) to \((-A,A)\). The way this is done
makes this mechanism meet the criterion (\ref{eq:neccsuff}) in Lemma
\ref{lem:neccsuff} with equality. \qedhere

\end{remark}

The following closed form condition is to our knowledge new.

\begin{theorem}\label{thm:logistic}

Let \(X\) be a standard Logistic variable having density
\(f_{\log}(x):=\frac{e^{-x}}{\left(1+e^{-x}\right)^{2}}\), and let \(q\)
be a real valued function on databases with global sensitivity
\(\Delta\). Then the mechanism that returns a variate of \(q(d) + s X\)
is \((\epsilon,\delta)\)-differentially private for \(\delta < 1\) if
and only if
\(s \geq s_{\log}(\epsilon, \delta, \Delta) := \frac{\Delta}{2 \log\left( \frac{e^{\frac{\epsilon}{2}} + \sqrt{\delta \left(e^{\epsilon}+\delta -1\right)}}{1-\delta}\right)}\).

\end{theorem}

\begin{remark}\label{rem:epsdp1l}

We have that
\(s_{1}(\epsilon, 0, \Delta) = s_{\log}(\epsilon, 0, \Delta) = \frac{\Delta}{\epsilon}\),
\(s_{1}(0, \delta, \Delta) = \frac{\Delta}{2 \log\left(\frac{1}{1-\delta} \right)}\),
and
\(s_{\log}(0, \delta, \Delta) = \frac{\Delta}{2 \log\left( \frac{1 + \delta}{1-\delta}\right)}\).
The Logistic mechanism serves as a second example besides the Laplace
mechanism that can achieve \(\epsilon\)-differential privacy. Both can
also achieve \((0, \delta)\)-differential privacy. Furthermore, letting
\(\epsilon_{1}(\epsilon, \delta) := \Delta / s_{1}(\epsilon, \delta, \Delta) = \epsilon - 2 \log\left(1-\delta \right)\),
and
\(\epsilon_{\log}(\epsilon, \delta) := \Delta / s_{\log}(\epsilon, \delta,\Delta) = 2 \log\left( \frac{e^{\frac{\epsilon}{2}} + \sqrt{\delta \left(e^{\epsilon}+\delta -1\right)}}{1-\delta}\right)\),
we get that for given \(\epsilon\) and \(\delta\), the Laplace and
Logistic mechanisms are \((\epsilon_{1}(\epsilon, \delta), 0)\)- and
\((\epsilon_{\log}(\epsilon, \delta), 0)\)-differentially private,
respectively. The difference between \(\epsilon_{1}(\epsilon, \delta)\)
and \(\epsilon_{\log}(\epsilon, \delta)\) can be interpreted as a
mechanisms sensitivity to \(\delta\). In practice, this translates to
the change in scale as \(\delta\) changes. \qedhere

\end{remark}

The following condition appears in Theorem 8 in
\citet{pmlr-v80-balle18a}.

\begin{theorem}\label{thm:balle}

Let \(Z\) be a standard Gaussian variable, and let \(q\) be a real
valued function on databases with global sensitivity \(\Delta\). Then
the mechanism that returns a variate of \(q(d) + \sigma Z\) is
\((\epsilon,\delta)\)-differentially private if and only if
\begin{align}
\label{eq:gaussian}
\Phi \left(\frac{\Delta}{2 \sigma}-\frac{\epsilon  \sigma}{\Delta}\right)-e^{\epsilon} \Phi \left(-\frac{\Delta}{2 \sigma}-\frac{\epsilon  \sigma}{\Delta}\right)\le \delta.
\end{align}

\end{theorem}

\hypertarget{sec:utility}{%
\subsection{\texorpdfstring{A benefit of achieving
\(\epsilon\)-differential
privacy}{A benefit of achieving \textbackslash epsilon-differential privacy}}\label{sec:utility}}

The well known Laplace and Gaussian mechanisms are both supported on
\(\ensuremath{\mathbb{R}}\), i.e., the noise densities are are supported
everywhere, which for us means that
\(\psi : \ensuremath{\mathbb{R}}\to \ensuremath{\mathbb{R}}\) as the
only time \(e^{-\psi} = 0\) is if \(\psi=\infty\). As we will see,
whether a symmetric log-concave mechanism supported everywhere can
achieve \(\epsilon\)-differential privacy or not has consequences for
its utility.

\begin{remark}\label{rem:epsdp}

In this work, we rely heavily on that for log concave densities \(f\),
the likelihood ratio \(f((z - \Delta)/s)/f(z/s)\) is monotone and
non-decreasing in \(z\). A mechanism that returns a variate
\(q(d) + sX\) where \(X\sim f\) can achieve
\((\epsilon, 0)\)-differential privacy when the likelihood ratio above
is bounded from above as \(z\) goes to infinity. The usual proof of
\(\epsilon\)-differential privacy for the Laplace mechanism essentially
shows that for \(s = \frac{\Delta}{\epsilon}\), this bound is
\(e^{\epsilon}\). \qedhere

\end{remark}

We now distinguish between those that are bounded and not bounded.

\begin{definition}[MLR-boundedness]\label{def:mlrunbounded}

Let \(\psi : \ensuremath{\mathbb{R}}\to \ensuremath{\mathbb{R}}\) be
convex and even, and let \(f(x)=e^{-\psi(x)}\). If for all
\(\Delta > 0\) and \(s > 0\) \[
\lim_{z\to\infty}\log
\left(
\frac{f(\frac{z - \Delta}{s})}{f(\frac{z}{s})}
\right) 
= 
\lim_{z\to\infty}\psi \left( \frac{z}{s} \right) - \psi \left(
\frac{z - \Delta}{s} \right) = \infty,
\] we say that \(\psi\) and \(f\) are monotone likelihood
ratio-unbounded (MLR-unbounded). If \(f\) is not MLR-unbounded it is
MLR-bounded.

\end{definition}

We can show the following.

\begin{theorem}\label{thm:unboundeds}

Let \(X\) be a random variable distributed according to density
\(f(x) = e^{-\psi(x)}\) where
\(\psi : \ensuremath{\mathbb{R}}\to \ensuremath{\mathbb{R}}\) is convex,
even, and MLR-unbounded. Then, \(s\to \infty\) as \(\delta \to 0\) for
the \((\epsilon, \delta)\)-differentially private mechanism that returns
a variate of \(q(d) + s X\).

\end{theorem}

\begin{remark}

Let \(M_{\delta=0}\) and \(M_{\delta>0}\) be symmetric log-concave
mechanisms with noise that is supported everywhere such that
\(M_{\delta=0}\) can achieve \(\epsilon\)-differential privacy and
\(M_{\delta>0}\) can only achieve \((\epsilon, \delta)\)-differential
privacy for \(\delta>0\). If we fix \(\epsilon\), then what Theorem
\ref{thm:unboundeds} means is that for any utility that is unbounded and
strictly decreasing in the scale of the added noise, we can choose a
small enough \(\delta\) such that the utility of \(M_{\delta=0}\) is
better by an arbitrary amount. In this sense, Theorem
\ref{thm:unboundeds} is a utility separation theorem for log-concave
mechanisms. \qedhere

\end{remark}

Specifically, for the Logistic and log-concave Subbotin densities,
\begin{align*}
\lim_{z \to \infty}\psi \left( \frac{z}{s}  \right) - \psi \left(
\frac{z - \Delta}{s} \right)
&=
\begin{cases}
\frac{\Delta}{s}, & \psi \in \{\psi_{\log},\psi_{1}\}, \\
\infty, & \psi \in \{\psi_{r} \mid r > 1\}
\end{cases} 
\end{align*} for \(s > 0\). This means that for fixed \(\epsilon\) we
can always choose \(\delta\) small enough such that the Laplace and
Logistic mechanisms have arbitrary smaller variances than mechanisms
that add \ensuremath{\text{Subbotin}_{r}} noise for \(r>1\), which
includes the Gaussian at \(r=2\).

\hypertarget{sec:comparison}{%
\subsection{Mechanism variance ratios}\label{sec:comparison}}

\figoneaistats

We now empirically investigate the last paragraph in the previous
Section. We choose the Laplace and Logistic mechanisms as
representatives of MLR-bounded mechanisms and the popular Gaussian
mechanism as the representative of the MLR-unbounded mechanisms that add
\ensuremath{\text{Subbotin}_{r}} noise.

Let \(X_{i} \sim f_{i}\) be a random variable distributed according to
symmetric log-concave density \(f_{i}\), and let
\(Y_{i} = q(d) + s_{i}X_{i}\) for
\(s_{i} = s_{i}(\epsilon, \delta, \Delta)\). For two mechanisms
\(Y_{a}\) and \(Y_{b}\), let
\(v_{a,b} := \frac{\operatorname{Var}(Y_{a})}{\operatorname{Var}(Y_{b})} = \left(\frac{s_{a}}{s_{b}} \right)^{2} \frac{\operatorname{Var}(X_{a})}{\operatorname{Var}(X_{b})}\).
Now, \(v_{a,b} \leq 1\) if and only if \begin{align}
\label{eq:rhocrit}
\rho_{a,b} := \frac{s_{a}}{s_{b}} \leq \sqrt{\frac{\operatorname{Var}(X_{b})}{\operatorname{Var}(X_{a})}}.
\end{align} From Lemma \ref{lem:standards} we get that for any
\(\Delta > 0\) \[
\rho_{a,b}(\epsilon, \delta) = \frac{s_{a}(\epsilon,
\delta, \Delta)}{s_{b}(\epsilon,
\delta, \Delta)} = \frac{s_{a}(\epsilon,
\delta, 1)}{s_{b}(\epsilon,
\delta, 1)}
\] for \(a, b \in \{\log\} \cup \{r \mid r \geq 1\}\). From condition
(\ref{eq:rhocrit}), we note that the curve given by
\(\rho_{a,b}(\epsilon, \delta) =\sqrt{\frac{\operatorname{Var}(X_{b})}{\operatorname{Var}(X_{a})}}\),
acts as a separator on the \((\epsilon, \delta)\)-space where one or the
other mechanism exhibit smaller variance. We will now compare the
selected mechanisms pairwise.

\hypertarget{sec:logvslap}{%
\subsubsection{Logistic vs.~Laplace mechanisms}\label{sec:logvslap}}

\figoneother

For completeness, we first compare the MLR-bounded mechanisms.

From Theorems \ref{thm:laplace} and \ref{thm:logistic} we can state
\(\rho_{1,\log}(\epsilon, \delta)\) in closed form. We start by
exploiting this to prove some specifics about
\(\rho_{1,\log}(\epsilon, \delta)\).

\begin{theorem}\label{thm:s1sl}

Let \(\Delta > 0\), \(\epsilon \geq 0\), \(1 > \delta \geq 0\), and
\begin{align*}
\rho_{1,\log}(\epsilon, \delta) := \frac{s_{1}(\epsilon, \delta, \Delta)}{s_{\log}(\epsilon, \delta,\Delta)}.
\end{align*} Then, \begin{align*}
\rho_{1,\log}&(\epsilon, \delta) \\&= \frac{2 \log\left(\sqrt{\delta} \sqrt{e^{\epsilon}+\delta -1}+e^{\frac{\epsilon}{2}}\right)-2 \log\left(1-\delta \right)}{\epsilon -2 \log\left(1-\delta \right)},
\end{align*} and \(2 > \rho_{1,\log}(\epsilon, \delta) \geq 1\) is
sharp.

\end{theorem}

An immediate consequence is the following.

\begin{corollary}\label{cor:v1vl}

Let \(X_{1}\) and \(X_{\log}\) be random variables distributed according
to \(f_{1}\) and \(f_{\log}\), respectively. Then, \[
2.44 > \frac{24}{\pi^{2}} > \frac{\operatorname{Var}(s_{1}(\epsilon, \delta,\Delta)X_{1})}{\operatorname{Var}(s_{\log}(\epsilon,
\delta,\Delta)X_{\log})} 
\geq \frac{6}{\pi^{2}} > 0.6,
\] and \[
\frac{\operatorname{Var}(s_{1}(\epsilon, \delta,\Delta)X_{1})}{\operatorname{Var}(s_{\log}(\epsilon,
\delta,\Delta)X_{\log})} \leq 1
\] when
\(\rho_{1,\log}(\epsilon, \delta) \leq \sqrt{\frac{\pi^{2}}{6}}\).

\end{corollary}

\begin{remark}\label{rem:compepspdl1.}

Combining Theorem \ref{thm:s1sl} with Remark \ref{rem:epsdp1l}, we see
that
\(\epsilon_{1}(\epsilon, \delta) \leq \epsilon_{\log}(\epsilon, \delta)\).
\qedhere

\end{remark}

A plot of \(v_{1, \log}(\epsilon, \delta)\) can be seen in Figure
\ref{fig:rat01}. The plot shows the contour line where
\(v_{1, \log}(\epsilon, \delta) = 1\). In the \(\epsilon\) and
\(\delta\) region where \(v_{1, \log}(\epsilon, \delta) < 1\), the
Laplace mechanism exhibits smaller variance than the Logistic mechanism.
From Corollary \ref{cor:v1vl}, we have that this region is determined by
\(\rho(\epsilon, \delta) < \sqrt{\frac{\pi^{2}}{6}}\). As an example,
for \(\delta \leq 10^{-4}\), the Laplace mechanism variance is smaller
as long as \(\epsilon \geq 0.0047\).

\hypertarget{sec:laploggauss}{%
\subsubsection{Laplace and Logistic vs.~Gaussian
mechanisms}\label{sec:laploggauss}}

When comparing Laplace and Logistic mechanisms above, we were able to
take advantage of \(s_{1}\) and \(s_{\log}\) having closed forms, and
consequently also \(\rho_{1,\log}\) having a closed form. Unfortunately,
we only have access to values for the Gaussian scale \(s_{2}\) through
numeric computations. The plots of \(v_{1,2}(\epsilon, \delta)\) and
\(v_{\log,2}(\epsilon, \delta)\) with the respective unit contours are
given in Figure \ref{fig:rat02} and Figure \ref{fig:rat12},
respectively. The regions where the Laplace and Logistic exhibit lower
variance are given by
\(\rho_{1,2}(\epsilon, \delta) < \sqrt{\frac{1}{2}}\) and
\(\rho_{\log,2}(\epsilon, \delta) < \frac{\sqrt{3}}{\pi}\),
respectively. In the plots, these are the regions where the respective
surfaces have values smaller than 1.

Numeric computations yield that for \(\epsilon \geq 0.05\), the Laplace
mechanism yields smaller mean squared error than both the Logistic and
the Gaussian mechanisms as long as \(\delta \leq 0.001\), and the
Logistic mechanism yields smaller squared error than the Gaussian
mechanism as long as \(\delta \leq 0.002\).

The variance is unbounded and strictly increasing in the scale for the
Logistic distribution as well as all
\ensuremath{\text{Subbotin}_{r\geq 1}} distributions. From Section
\ref{sec:utility} we have that since both the Laplace and Logistic noise
is MLR-bounded, while the Gaussian is MLR-unbounded, the ratios
\(v_{1,2}\) and \(v_{\log, 2}\) go to 0 as \(\delta \to 0\) for any
\(\epsilon\). Furthermore, for all symmetric densities, we have that the
mechanism that add noise according to these exhibit a mean squared error
(MSE) that is equivalent to the variance. Consequently, the variance
comparisons above hold for the MSE as well. This is particularly
relevant for histogram queries, since each individual histogram count
can be computed using a one-dimensional mechanism.

\hypertarget{sec:extensions}{%
\section{THE MULTIDIMENSIONAL CASE}\label{sec:extensions}}

So far, we discussed \((\epsilon, \delta)\)-differential privacy for
real-valued database queries \(q\), i.e.,
\(q(d)\in \ensuremath{\mathbb{R}}\) for database \(d\). In order to
generalize our results regarding one-dimensional mechanisms to queries
in \(\ensuremath{\mathbb{R}}^{n}\), we generalize a symmetric
log-concave density \(f\) to a \(\|\cdot\|\)-spherically symmetric
log-concave density \(C f(\|x\|)\) for a normalization constant
\(C \in \ensuremath{\mathbb{R}}\). For this generalization, the
isodensity surfaces are spherical with respect to the norm
\(\|\cdot\|\). If \(n=1\), this density coincides with \(f\) and we will
simply write \(f\) also for \(n>1\).

Furthermore, in our above discussion of adding real-valued noise, a
critical parameter was the global sensitivity \(\Delta\). This maximal
change was defined in terms of the absolute value \(|\cdot|\). In the
multidimensional case we generalize this as \begin{align*}
\Delta := \max_{(d,d') \in \ensuremath{\mathcal{N}}} \|q(d) - q(d')\|
\end{align*} where \(\|\cdot\|\) is a norm in
\(\ensuremath{\mathbb{R}}^{n}\).

Our second main result comes in two parts. The first of these, Lemma
\ref{lem:iso} below, essentially extends Lemma \ref{lem:neccsuff} to the
case of \(\|\cdot\|\)-spherically symmetric log-concave densities when
the global sensitivity is defined using the same norm \(\|\cdot\|\). The
second part, Theorem \ref{thm:gensub} below, in turn specializes Lemma
\ref{lem:iso} to a case where the noise can be represented by a vector
of iid random variables.

\begin{lemma}\label{lem:iso}

Let \(\ensuremath{\mathbf{X}} \in \ensuremath{\mathbb{R}}^{n}\) for
\(n\geq 1\) be distributed according to \(\|\cdot\|\)-spherically
symmetric log-concave density \(f\). Then, if \(q\) and \(w\) are
\(\ensuremath{\mathbb{R}}^{n}\) and \(\ensuremath{\mathbb{R}}\) valued
functions on databases, respectively, both with global sensitivity
\(\Delta = \max_{(d, d') \in \ensuremath{\mathcal{N}}}\| q(d) - q(d')\| = \max_{(d, d') \in \ensuremath{\mathcal{N}}} |w(d) - w(d')|\),
then for \(s > 0\), the mechanism \(q(d) + s\ensuremath{\mathbf{X}}\) is
\((\epsilon, \delta)\)-differentially private if and only if the
mechanism \(w(q) + sY\) for \(Y \sim f\) is.

\end{lemma}

Let
\(\|(v_{1}, v_{2}, \ldots, v_{n})\|_{p} := (\sum_{i=1}^{n} |v_{i}|^{p})^{1/p}\)
for \(p \in \ensuremath{\mathbb{R}}\), \(p\geq 1\) denote the \(p\)-norm
on \(\ensuremath{\mathbb{R}}^{n}\). Using particulars of the
multivariate Gaussian density, \citet{pmlr-v80-balle18a} show that for
\(\Delta\) defined in terms of the Euclidean norm \(\|\cdot\|_{2}\), the
condition (\ref{eq:gaussian}) is sufficient and necessary for the
mechanism that adds \(s (Z_{1}, Z_{2}, \ldots, Z_{n})\) where the
\(Z_{i}\) are independent standard Gaussian variables.
\citet{10.2307/2337215} show that when the independent
\(X_{i} \sim f_{p}\), where \(f_{p}\) is the
\ensuremath{\text{Subbotin}_{p}} density, the variable
\((X_{1}, X_{2}, \ldots, X_{n})\) has a \(\|\cdot\|_{p}\)-spherical
distribution. Using this result, we specialize Lemma \ref{lem:iso} to
pairings of \ensuremath{\text{Subbotin}_{p}} variables and \(p\)-norms
for \(p\geq 1\) in the following Theorem.

\begin{theorem}\label{thm:gensub}

Let \(n \in \ensuremath{\mathbb{N}}\), \(n \geq 1\), \(p\geq 1\),
\(X_{1}, \ldots, X_{n}\) be independent \ensuremath{\text{Subbotin}_{p}}
random variables, and
\(\ensuremath{\mathbf{X}} = (X_{1}, X_{2}, \ldots, X_{n})\). Then for a
\(\ensuremath{\mathbb{R}}^{n}\)-valued function \(q\) on databases with
global sensitivity
\(\Delta = \max_{(d, d') \in \ensuremath{\mathcal{N}}}\| q(d) - q(d')\|_{p}\)
and a database \(d\), the mechanism returning a variate of
\(q(d) + s \ensuremath{\mathbf{X}}\) is
\((\epsilon, \delta)\)-differentially private if and only if
\begin{align*}
F_{p} \left( \frac{\Delta - t}{s} \right) - 
e^{\epsilon}F_{p} \left( -\frac{t}{s} \right)
\leq \delta 
\end{align*} where \begin{align*}
t := \sup \left\{z \in \ensuremath{\mathbb{R}}\mid \psi_{p}\left(\frac{z}{s}\right) -
\psi_{p}\left(\frac{z - \Delta}{s} \right) \leq \epsilon\right\}.
\end{align*} Furthermore, if the above holds for scale \(s > 0\), then
it also holds for scales \(s' > s\).

\end{theorem}

We will call the mechanism from Theorem \ref{thm:gensub} the
\emph{\ensuremath{\text{Subbotin}_{r}} mechanism} and let
\(s_{r}(\epsilon, \delta, \Delta)\) denote the minimum scale \(s\) for
which the condition in Theorem \ref{thm:gensub} holds for the
\ensuremath{\text{Subbotin}_{r}} mechanism.

\figtwoaistats

\begin{remark}\label{rem:laplace}

Combining Theorem \ref{thm:gensub} with Theorem \ref{thm:laplace}, we
get that the vector Laplace mechanism
\(q(d) + (s X_{1}, s X_{2}, \ldots, s X_{n})\) where the independent
\(X_{i} \sim f_{1}\) and the global sensitivity \(\Delta\) is defined in
terms of \(\| \cdot \|_{1}\) is \((\epsilon, \delta)\)-differentially
private if and only if \[
s \geq \frac{\Delta}{\epsilon - 2 \log\left(1-\delta \right)}. \qedhere
\]

\end{remark}

\hypertarget{sec:subopt}{%
\subsection{Optimizing query responses using subbotin
noise}\label{sec:subopt}}

Let \(R\) be a set of record values let
\(q : R^{n} \to \ensuremath{\mathbb{R}}^{m}\) be a query function on
length \(n\) databases \(d\). An important problem is to minimize the
MSE \(m^{-1} \operatorname{E}[\|M(d) - q(d)\|_{2}^{2}]\) for
multidimensional mechanisms \(M\). Given a choice between different
mechanisms, a question is then which one minimizes the MSE. We now look
at this question for the family of \ensuremath{\text{Subbotin}_{r}}
mechanisms \(M_{r}\). Let \(\Delta_{r}\) denote the global sensitivity
of \(q\) defined using the \(r\)-norm. Then, using definitions from
Theorem \ref{thm:gensub},
\(m^{-1}\operatorname{E}[\|M_{r}(d) - q(d)\|_{2}^{2}] = m^{-1}\operatorname{E}[\|s\ensuremath{\mathbf{X}}\|_{2}^{2}] = s^{2}\operatorname{Var}(X)\)
for \(X \sim f_{r}\) and \(s = s_{r}(\epsilon, \delta,\Delta_{r})\). We
can minimize this over the choices of \(r\) by letting \begin{align}
\label{eq:pmin}
r = \arg\min_{p\geq 1}  s_{p}(\epsilon,
\delta, \Delta_{p})^{2} \left( {p}^{(2/p)}\frac{\Gamma(3/p)}{\Gamma(1/p)}  \right).
\end{align} The optimization (\ref{eq:pmin}) requires determining
\(\Delta_{r}\). This can be difficult. For particular classes of
queries, we can make it a little easier. We call query function \(q\)
\emph{linear} if there exist linear function
\(f : R \to \ensuremath{\mathbb{R}}^{m}\) and function
\(\nu : \ensuremath{\mathbb{N}}\to \ensuremath{\mathbb{R}}\) such that
\(q(r_{1}, r_{2}, \ldots, r_{n}) = \nu(n) \sum_{i=1}^{n} f(r_{i})\). To
clarify which \(f\) and \(\nu\) we associate with linear \(q\), we write
\(q_{\nu, f}\). Let
\(\|S\|_{\infty} := \sup\{\| x - y \|_{\infty} \mid x,y \in S\}\) denote
the \(l_{\infty}\) diameter of a set
\(S \subseteq \ensuremath{\mathbb{R}}^{m}\). Then, if
\(\|f(R)\|_{\infty}\) is easier to determine or bound than
\(\Delta_{r}\) the following can help.

\begin{theorem}\label{thm:inorm}

Let \(q_{\nu,f} : R^{n} \to \ensuremath{\mathbb{R}}^{m}\) be a linear
database query function for some integers \(m>0\) and \(n>0\). Then
\(\Delta_{p} \leq m^{1/p} \,|\nu(n)|\,\|f(R)\|_{\infty}\) with equality
if \(f(R)\) is an \(l_{\infty}\)-ball.

\end{theorem}

A natural evaluation of optimizing (\ref{eq:pmin}) is to compare with
the go-to Gaussian mechanism. We therefore essentially replicate the
estimation of the mean \(m\)-dimensional vector experiment performed by
\citet{pmlr-v80-balle18a}. We also chose this experiment because Balle
and Wang use this to evaluate two de-noising post-processing methods for
the Gaussian mechanism that we also want to compare with.

\figtwoother

Let the mean vector query be the linear query \(q = q_{\nu,f}\) for
\(f(x) := x\) and \(\nu(n) := 1/n\). Suppose \(y\) is the output of the
Gaussian mechanism with scale \(s\) for query \(q\). Then, the first of
the de-noising methods is the adaptive James-Stein estimator
\(y_{JS} := \left(1 - \frac{(m - 2)s^{2}}{\|y\|_{2}^{2}} \right)\), with
the underlying assumption that the query response is a random Gaussian
vector with iid dimensions. The second is the soft-thresholding proposed
by Donoho et al.
\citetext{\citeyear{10.1093/biomet/81.3.425}; \citeyear{donohoDenoisingSoftthresholding1995}}
given by \(y_{th} := \sign(y)\max(0, |y| - t)\) where
\(t := s \sqrt{2\log(m)}\). For this method, ideally \(t\) is chosen
such that it defines an \(\infty\)-ball containing the noise. We
recognize \(t\) as a standard upper bound for the expected maximum of
\(m\) iid scale \(s\) Gaussian random variables, so the defined
\(\infty\)-ball contains the noise in expectation. We adapt the soft
thresholding to scale \(s\) \ensuremath{\text{Subbotin}_{r}} mechanism
output by letting \(t = s/300 \sum_{i=1}^{300} \max(x_{i})\) where the
\(m\) elements in \(x_{i} \in \ensuremath{\mathbb{R}}^{m}\) are
independent standard \ensuremath{\text{Subbotin}_{r}} variates.

Like Balle and Wang, we require that records lie in the unit
\(l_{\infty}\)-ball \(R = v + [-1/2,1/2]\), where the \(v_{i}\) in
\(v = (v_{1}, v_{2}, \ldots, v_{m})\) are standard normal variates, let
databases be length \(n=500\) sequences of elements from \(R\) chosen
uniformly at random, and performed our experiments with fixed
\(\delta = 10^{-4}\) and errors measured by the \(l_{2}\)-error. For
each pair \((\epsilon, m)\) in
\(\{0.01, 0.1, 1\} \times \{10, 100, 500, 1000,2000\}\) we numerically
optimized (\ref{eq:pmin}) on a regular grid \((1,1.5, 2,\ldots, 14)\)
yielding \(r_{\epsilon, m}\) and generated 100 databases. For each of
the databases from we then computed the \(l_{2}\)-error of the Gaussian
mechanism (Gauss), the \(\ensuremath{\text{Subbotin}_{r_{\epsilon,m}}}\)
mechanism (Sub(r)), their soft-thresholded versions (Gauss-t, Sub(r)-t),
the James-Stein denoised Gaussian mechanism (Gauss-JS), after which we
computed the average for each of the five mechanisms for the 100
databases. These averages can be seen in Figures \ref{fig:e1err} --
\ref{fig:e001err}. In Figure \ref{fig:e001err} the Gauss errors were
much larger than all the others and were truncated to allow details to
be seen for the other error series. Figure \ref{fig:vars} shows the mean
squared errors \(s^{2}\operatorname{Var}(X)\) computed during the
optimization of (\ref{eq:pmin}) for \(\epsilon=1\) and all values \(m\).

Since the optimal scale for the \ensuremath{\text{Subbotin}_{r}}
mechanisms is generally not available analytically, neither is the
optimal \(r\) as a function of \(m\). Therefore, we optimized \(r\)
numerically on a grid of \(r\) values as described above. The maximum
grid value \(r=14\) was chosen to contain a local optimum for all \(m\),
which are listed below. The grid points are shown as dots in Figure
\ref{fig:vars}.

For \(\epsilon=1\), \(\epsilon=0.1\), and \(\epsilon=0.01\), the
selected \(r_{\epsilon, m}\) values were (2,4,6,7,7.5), (2.5, 5, 7.5,
8.5, 9), and (3.5, 7, 10.5, 11.5, 13), respectively. The corresponding
\ensuremath{\text{Subbotin}_{r}} scales rounded to two decimal places
were (0.02, 0.06, 0.08, 0.09, 0.10), (0.16, 0.37, 0.52, 0.58, 0.63), and
(1.14, 2.07, 2.63, 2.84, 3.04), respectively. The corresponding Gaussian
mechanism scales were (0.02, 0.06, 0.14, 0.20, 0.28), (0.16, 0.49, 1.10,
1.55, 2.19), and (1.09, 3.45, 7.72, 10.91, 15.44), respectively.

We make the following observations. Optimizing (\ref{eq:pmin}) is
relevant for minimizing the \(l_{2}\)-error because of the monotone
relationship between the MSE and the \(l_{2}\)-error. Doing so
consistently produced smaller \(l_{2}\)-errors over the Gaussian
mechanism and its denoised versions. It is well known that
\(\frac{m}{\sqrt{m+1}} \leq \operatorname{E}[\|\ensuremath{\mathbf{Z}}\|_{2}] \leq \sqrt{m}\)
for a length \(m\) vector \(\ensuremath{\mathbf{Z}}\) of iid standard
Gaussian random variables. Also knowing that
\(\Delta \propto \sqrt{m}\), we can apply Lemma \ref{lem:standards} to
explain why the \(l_{2}\)-error of the Gaussian mechanism in our
experiments appears linear in the query dimension \(m\). This
application of Lemma \ref{lem:standards} may be of independent interest.
The James-Stein estimator always improved on the unprocessed Gaussian
mechanism (we refer the reader to \citet{pmlr-v80-balle18a} for an
analysis). The soft-thresholded mechanisms were better than their
unprocessed counterparts when the scale was large, while the opposite
was true when the scale was smaller. Empirically, for Gaussian mechanism
in our experimental setting, the changeover happens close to scale
\(s=1\), for which the \(v\) and noise distributions are equal. When the
scale becomes large as \(\epsilon\) becomes small and \(m\) large,
almost all of the entries in the thresholded estimates are 0, which
suggests why the errors of the thresholded Gaussian and
\(\ensuremath{\text{Subbotin}_{r_{\epsilon,m}}}\) mechanisms for
\(\epsilon=0.01\) are very similar and their comparison might be less
informative. Also, when considering the results, it might be helpful to
remember that we chose the experimental setup to favor the James-Stein
estimator for Gaussian noise.

\hypertarget{sec:background}{%
\section{BACKGROUND AND OTHER RELATED WORKS}\label{sec:background}}

Adding randomness to the computation of query results in order to
protect privacy has a long history (see e.g.,
\citet{dalenius_information_1978}, Section 13 or
\citet{denningSecureStatisticalDatabases1980}). The definitions of
privacy in these works differed and were often concerned with
disallowing the variance in the query answers resulting from the
randomness to be reduced too much by repeated queries. Following work by
\citet{dinurRevealingInformationPreserving2003}, noting that a
``succinct catch-all definition of privacy is very elusive'', Dwork et
al.~published \citeyearpar{Dwork2006a} a definition of privacy that
would become known as \(\epsilon\)-differential privacy. This definition
has three important properties: quantification of a strong form of
privacy, can be composed across queries types and databases, and is
immune to post-processing. The latter meaning that no amount of
processing after the fact will change the guarantees. The definition of
\(\epsilon\)-differential privacy was soon extended to
\((\epsilon, \delta)\)-differential privacy
\citep{our-data-ourselves-privacy-via-distributed-noise-generation},
which allowed the use of, among others, output perturbation by Gaussian
noise, and turned out to improve the compositional properties
\citep{dworkBoostingDifferentialPrivacy2010, 7883827} and thereby the
utility--privacy trade-off over multiple queries. However, the extension
to \((\epsilon, \delta)\)-differential privacy represents a weakening of
the privacy guaranteed over pure \(\epsilon\)-differential privacy that
is not uncontroversial \citep{mcsherryHowManySecrets2017}. In a machine
learning context, compositional properties received much attention for
establishing end-to-end privacy, particularly when using tools for
private optimization
\citep{10.1109/FOCS.2014.56, abadiDeepLearningDifferential2016}.
Partially as a result of this, several similar relaxations of
\(\epsilon\)-differential privacy have been proposed that improve
compositional properties over their predecessors. These include
Concentrated differential privacy
\citep{DBLP:journals/corr/DworkR16, 10.1007/978-3-662-53641-4_24},
\text{R\'enyi} differential privacy
\citep{mironovRenyiDifferentialPrivacy2017}, and Gaussian differential
privacy \citep{DBLP:journals/corr/abs-1905-02383}. Similarly to the work
of Balle et al.~\citeyearpar{pmlr-v80-balle18a}, our work can be seen as
complementary to these efforts in that we concentrate on optimizing
privacy parameters for \emph{single} applications of a particular class
of mechanisms under \((\epsilon, \delta)\)-differential privacy.
\citet{8375673} defines the Generalized Gaussian mechanism that is
effectively equivalent to the \ensuremath{\text{Subbotin}_{r}}
mechanism. However, only a sufficient criterion for
\((\epsilon, \delta)\)-differential privacy for this family is
developed. For the Gaussian case, this criterion reduces to a criterion
inferior to the one described by \citet{6606817}. They also do not
consider a systematic search among members in this family for optimizing
utility. Optimality with respect to different utilities has been shown
for \(\epsilon\)-differentially private mechanisms such as the staircase
mechanism \citep{7093132} and randomized response in a multi-party
setting \citep{kairouzSecureMultipartyDifferential2015}.

\hypertarget{sec:discussion}{%
\section{SUMMARY AND DISCUSSION}\label{sec:discussion}}

We provided a necessary and sufficient condition for
\((\epsilon, \delta)\)-differential privacy for mechanisms that add
noise from a symmetric and log-concave distribution (Lemma
\ref{lem:neccsuff}). This condition is given directly in terms of the
standard distribution function, the needed scale of the distribution,
\(\epsilon\), \(\delta\), and the global sensitivity \(\Delta\) of the
query in question. Previously, such a condition was only known for the
Gaussian distribution \citep{pmlr-v80-balle18a} and was shown to be
essentially identical for the discrete Gaussian distribution
\citep{NEURIPS2020_b53b3a3d}.

Importantly, we show that our necessary and sufficient condition extends
naturally to multidimensional query responses. Key to this extension is
that the noise distribution is spherically symmetric log-concave and is
spherical with respect to the norm used to define the query global
sensitivity. We are able to show that this holds for mechanisms where
the noise vector consists of independent
\ensuremath{\text{Subbotin}_{r}} random variables and the query
sensitivity is assessed using the corresponding \(l_{r}\)-norm for
\(r\geq 1\) (Theorem \ref{thm:gensub}). This infinite family of
\ensuremath{\text{Subbotin}_{r}} mechanisms contains the prototypical
Laplace and Gaussian mechanisms with \(r=1\) and \(r=2\), respectively.

In the one-dimensional query case, we show that for symmetric
log-concave mechanisms supported everywhere that cannot achieve pure
\(\epsilon\)-differential privacy, the scale grows to infinity as
\(\delta\) goes to 0 (Theorem \ref{thm:unboundeds}). Since this is not
the case when the mechanism can achieve \(\epsilon\)-differential
privacy, we achieve a separation of these two types of symmetric
log-concave mechanism for any utility that is decreasing and unbounded
in the noise scale. We demonstrate this utility separation by showing
that the Laplace and Logistic mechanisms exhibits smaller variance than
the Gaussian mechanism for a large and relevant range of values for
\(\epsilon\) and \(\delta\).

We further showed that the optimal scale for log-concave noise
mechanisms is proportional to the query sensitivity \(\Delta\) (Lemma
\ref{lem:standards}). This means that the optimal noise scale is
proportional to the behaviour of the norm used for the query
sensitivity. As a consequence of this and having access to a
\ensuremath{\text{Subbotin}_{r}} mechanism for every \(l_{r}\)-norm for
\(r\geq 1\), we were able to demonstrate that optimizing the choice of
\(r\) for the multidimensional case can significantly improve the mean
squared error and any similar utility such as the \(l_{2}\)-error over
any fixed choice of mechanism such as the Gaussian.

Consider that a high-dimensional random vector
\(\ensuremath{\mathbf{X}}\) with iid components appears to be
concentrated on a sphere (see e.g., \citet{biau:hal-00879436}), a
phenomenon that \citet{Dong2021CLTDP} state as that the mechanism adding
\(\ensuremath{\mathbf{X}}\) essentially behaves like a Gaussian
mechanism. So, for a fixed high-dimensional query and given privacy
parameters \(\epsilon\) and \(\delta\), two different Subbotin
mechanisms can essentially behave like two Gaussian mechanisms that
differ in the noise scale. This suggests a ``Gaussian lens'' through
which to interpret a choice of norm. However, such an interpretation
cannot be strict: even though the Laplace mechanism essentially behaves
like a Gaussian mechanism in high dimensions, it differs in the type of
privacy guarantee it affords since the Gaussian mechanism cannot achieve
\((\epsilon, 0)\)-differential privacy for any \(\epsilon\).

\citet{dworkAlgorithmicFoundationsDifferential2014} warn that claims of
differential privacy should be seen in relation to the granularity of
data, essentially reflected in the neighborhood relation used to define
privacy. Similarly, the choice of norm affects sensitivity and the
warning could be considered to apply here as well, particularly in light
of the preceding paragraph. We recommend taking this into account when,
for example, considering the range of \(r\) to optimize Subbotin
mechanisms over.

As mentioned in Section \ref{sec:background}, other members in the
family of differential privacy definitions exhibit better composition
properties and can be attractive alternatives to
\((\epsilon, \delta)\)-differential privacy. They have in common that
the query sensitivity is defined by some norm and that an increase in
the sensitivity yields an increase in the scale of the noise. For
example, in the Concentrated differential privacy case, the sensitivity
enters the scale of the Gaussian mechanism as a multiplicative factor
for a given level of privacy. Furthermore, the extension of a scale
bound for univariate unimodal and symmetric distributions to the
multidimensional spherical case appears quite general. Therefore, we
consider our results regarding \ensuremath{\text{Subbotin}_{r}}
mechanisms and optimizing the choice of distribution in terms of query
response dimensionality as an example of how to implement a general
method under \((\epsilon, \delta)\)-differential privacy.

\putacknowledgement

\bibliography{bibliography}
\renewcommand{\bibliography}[1]{}

\startsupp

\hypertarget{app:proofs}{%
\section{POSTPONED PROOFS}\label{app:proofs}}

After recapitulating some known material, we present postponed proofs
section-wise.

\hypertarget{app:prelims}{%
\subsection{Important known results we use}\label{app:prelims}}

For completeness, we include the proofs of the following known results
we rely on. Their presentation is tailored to our needs.

We will be using the following result repeatedly. The formulation and
proof is taken from \citet{saumard2014logconcavity}.

\begin{lemma}[Monotone likelihood ratio]\label{lem:mlr}

A density \(f\) on \(\ensuremath{\mathbb{R}}\) is log-concave if and
only if the translation family
\(\{f(\cdot - \theta) \mid \theta \in \ensuremath{\mathbb{R}}\}\) has
monotone likelihood ratios: i.e., for every \(\theta \leq \theta'\) the
ratio \(f(x-\theta')/f(x - \theta)\) is a monotone nondecreasing
function of \(x\).

\end{lemma}

\begin{proof}[\bf\em Proof of Lemma \ref{lem:mlr}]

First note that for all \(x < x'\) and \(\theta < \theta'\)
\begin{align}
\label{eq:mlrdef}
\frac{f(x-\theta')}{f(x-\theta)} &\leq \frac{f(x'-\theta')}{f(x'-\theta)}
\end{align} iff \begin{align}
\label{eq:logc1}
\log f(x-\theta') + \log f(x'-\theta) \leq 
\log f(x'-\theta') + \log f(x-\theta).
\end{align} Let \(t = (x' - x)/(x' - x + \theta' - \theta)\) and note
that \begin{align*}
x - \theta &= t (x - \theta') + (1 - t) (x' - \theta) \\
x' - \theta' &= (1 - t) (x - \theta') + t (x' - \theta). 
\end{align*} Log-concavity of \(f\) implies that \begin{align*}
\log f(x - \theta) &\geq t \log f(x - \theta') + (1-t) \log f(x' -
\theta) \\
\log f(x' - \theta') &\geq (1-t) \log f(x - \theta') + t \log f(x' -
\theta).
\end{align*} Adding these yields (\ref{eq:logc1}), and we can conclude
that \(f\) being log-concave implies (\ref{eq:mlrdef}).

Now, suppose that (\ref{eq:mlrdef}) holds. Then, (\ref{eq:logc1}) holds,
and does so in particular if \(x,x', \theta, \theta'\) satisfy
\(x - \theta' = a < b = x' - \theta\) and
\(t = (x' - x)/(x' - x + \theta' - \theta) = 1/2\), so that
\(x-\theta = (a + b)/2 = x' - \theta'\). Then (\ref{eq:logc1}) becomes
\[
\log f(a) + \log f(b) \leq 2 \log f((a + b)/2).
\] This, together with measurability of \(f\), implies that \(f\) is
log-concave. \qedhere

\end{proof}

\begin{lemma}[Translation invariance]\label{lem:translation}

Let \(X\) be a random variable taking values in
\(\ensuremath{\mathbb{R}}^{n}\) for natural number \(n \geq 1\). For
arbitrary but fixed \(x,y,z \in \ensuremath{\mathbb{R}}^{n}\) we have
that
\(\Pr(X + (x-z) \in S) \leq e^{\epsilon} \Pr(X + (y - z)\in S) + \delta\)
for all measurable \(S \subseteq \ensuremath{\mathbb{R}}^{n}\) implies
\(\Pr(X + x) \in S) \leq e^{\epsilon} \Pr(X + y \in S) + \delta\) for
all measurable \(S \subseteq \ensuremath{\mathbb{R}}^{n}\).

\end{lemma}

\begin{proof}[\bf\em Proof of Lemma \ref{lem:translation}]\label{proof:lem:translation}

Follows directly from that if \(S\subseteq \ensuremath{\mathbb{R}}^{n}\)
is measurable, so is \(r + S\) for any
\(r \in \ensuremath{\mathbb{R}}^{n}\), including \(r = -z\). \qedhere

\end{proof}

\begin{lemma}[Symmetry of distance]\label{lem:symmetry}

Let \(f\) be an even density and let \(X \sim f\), let
\(\ensuremath{\mathcal{S}}\) be the measureable sets
\(S \subseteq \ensuremath{\mathbb{R}}^{n}\) for some positive integer
\(n\), \(d \geq 0\), \(\epsilon \geq 0\), \(\delta \geq 0\). Then

\begin{align*}
(\forall S \in \ensuremath{\mathcal{S}}) \Pr(d + X \in S) &- e^{\epsilon}\Pr(X \in S)
\leq \delta \\
&\iff \\
(\forall S \in \ensuremath{\mathcal{S}}) \Pr(-d + X \in S) &- e^{\epsilon}\Pr(X \in S)
\leq \delta.
\end{align*}

\end{lemma}

\begin{proof}[\bf\em Proof of Lemma \ref{lem:symmetry}]\label{proof:lem:symmetry}

Note that for \(S \in \ensuremath{\mathcal{S}}\) \begin{align}
\label{eq:symmproofa}
S \in \ensuremath{\mathcal{S}} &\iff -S \in \ensuremath{\mathcal{S}} \\
\label{eq:symmproofb}
x \in S &\iff -x \in -S \\
\label{eq:symmproofc}
X &\stackrel{d}{=} -X
\end{align} where \(\stackrel{d}{=}\) means equal in distribution. Now,
\begin{align*}
(\forall S \in \ensuremath{\mathcal{S}}) \Pr(d + X \in S) &- e^{\epsilon}\Pr(X \in S)
\leq \delta \\
&\stackrel{(\text{by (\ref{eq:symmproofa}))}}{\iff} \\
(\forall S \in \ensuremath{\mathcal{S}}) \Pr(d + X \in -S) &- e^{\epsilon}\Pr(X \in -S)
\leq \delta \\
&\stackrel{(\text{by (\ref{eq:symmproofb}))}}{\iff} \\
(\forall S \in \ensuremath{\mathcal{S}}) \Pr(-(d + X) \in S) &- e^{\epsilon}\Pr(-X \in S)
\leq \delta \\
&\stackrel{(\text{by (\ref{eq:symmproofc}))}}{\iff} \\
(\forall S \in \ensuremath{\mathcal{S}}) \Pr(-d + X \in S) &- e^{\epsilon}\Pr(X \in S)
\leq \delta. \qedhere
\end{align*}

\end{proof}

The result below appears in the proof of Theorem 5 in
\citet{pmlr-v80-balle18a}.

\begin{lemma}[Removal of quantifier]\label{lem:equiv}

Let random variables \(Y_{+}\) and \(Y_{-}\) taking values from
\(\ensuremath{\mathbb{R}}^{n}\) be distributed according to densities
\(f_{+}\) and \(f_{-}\), respectively. Let \(\ensuremath{\mathcal{S}}\)
denote the collection of all measurable sets of elements in \(V\), and
let \(A^{c} = \{x \mid f_{+}(x) > e^{\epsilon} f_{-}(x) \}\). Then,
\begin{align*}
\forall S \in \ensuremath{\mathcal{S}} \; \Pr(Y_{+} \in S) &- e^{\epsilon} \Pr(Y_{-} \in S) \leq \delta \\
&\iff\\
\Pr(Y_{+} \in A^{c}) &- e^{\epsilon} \Pr(Y_{-} \in A^{c}) \leq \delta.
\end{align*}

\end{lemma}

\begin{proof}[\bf\em Proof of Lemma \ref{lem:equiv}]\label{proof:lem:equiv}

Let \(A = \{x \mid f_{+}(x) \leq e^{\epsilon} f_{-}(x) \}\) and note
that \(A^{c}\) is the complement of \(A\). Then, \begin{align*}
\Pr(Y_{+} \in S \cap A) &= \int_{S \cap A}f_{+}(x) dx \\
&\leq \int_{S \cap A} e^{\epsilon}f_{-}(x) dx \\
&= e^{\epsilon}\Pr(Y_{-} \in S \cap A).
\end{align*} Using this, we get \begin{align*}
\Pr(Y_{+} &\in S) - e^{\epsilon} \Pr(Y_{-} \in S) \\
&= \Pr(Y_{+} \in S \cap A) + \Pr(Y_{+} \in S \cap A^{c}) \\
&\phantom{=}\, - e^{\epsilon} \left( \Pr(Y_{-} \in S \cap A) + \Pr(Y_{-} \in S \cap
A^{c}) \right) \\
&\leq \Pr(Y_{+} \in S \cap A^{c}) - e^{\epsilon}\Pr(Y_{-} \in S \cap
A^{c})\\
&\leq \int_{S \cap A^{c}} \left( f_{+}(x) - e^{\epsilon}f_{-}(x) \right) dx\\
&\leq \int_{A^{c}} \left( f_{+}(x) - e^{\epsilon}f_{-}(x) \right) dx\\
&= \Pr(Y_{+} \in A^{c}) - e^{\epsilon}\Pr(Y_{-} \in A^{c}).
\end{align*} The last inequality comes from
\(f_{+}(x) > e^{\epsilon}f_{-}(x)\) for \(x \in A^{c}\), so
\(f_{+}(x) - e^{\epsilon}f_{-}(x) \geq 0\) for \(x \in A^{c}\).

Since
\(\Pr(Y_{+} \in S) - e^{\epsilon} \Pr(Y_{-} \in S) \leq \Pr(Y_{+} \in A^{c}) - e^{\epsilon}\Pr(Y_{-} \in A^{c})\)
the Lemma follows from
\(\Pr(Y_{+} \in A^{c}) - e^{\epsilon}\Pr(Y_{-} \in A^{c}) \leq \delta\)
implies
\(\Pr(Y_{+} \in S) - e^{\epsilon} \Pr(Y_{-} \in S) \leq \delta\).
\qedhere

\end{proof}

\hypertarget{app:b}{%
\subsection{\texorpdfstring{Proofs from Section
\ref{sec:contrib}}{Proofs from Section }}\label{app:b}}

This section contains the proof of the first of our main results given
in Lemma \ref{lem:neccsuff} regarding the privacy of the mechnism
\(M(d) = q(d) + sX\).

The proof can be outlined as follows. Letting
\(\ensuremath{\mathcal{S}}\) denote the measureable sets, the first step
in the proof is the removal of the quantifier over
\(\ensuremath{\mathcal{S}}\) in the privacy criterion \begin{align}
\label{eq:crite}
(\forall S \in \ensuremath{\mathcal{S}}) \Pr(M(d) \in S) - e^{\epsilon}\Pr(M(d') \in S) \leq \delta
\end{align} given by Lemma \ref{lem:equiv} above. This removal shows
that that there exists a partition of the sample space into regions
\(A\) and \(A^{c}\) such that \[
\Pr(M(d) \in S) > e^{\epsilon} \Pr(M(d') \in S)
\] if and only if \(S \cap A^{c} \neq \emptyset\). The next step uses
log-concavity of the density \(f\) of \(X\), particularly the
monotonicity of likelihood ratios in the translation family of \(f\)
given in Lemma \ref{lem:mlr}, to show in Lemma \ref{lem:probscrit} that
we can treat \(A^{c}\) as an interval \((t, \infty)\) where threshold
\(t\) depends on \(\epsilon\), \(s\) and \(d=|q(d) - q(d')|\). This
allows stating the criterion (\ref{eq:crite}) above in terms of \[
\Pr(d + sX > t) - e^{\epsilon}\Pr(sX > t) \leq \delta,
\] which is equivalent to the statement in Lemma \ref{lem:neccsuff}. The
proof of Lemma \ref{lem:neccsuff} is completed by showing monotonicity
of the right hand side in the above display in Lemma
\ref{lem:monotonicity}.

The above steps are shown using elementary means and their
demonstrations are relatively self-contained, which we consider a
benefit. Alternatively, machinery based on conjugate duality and
hypothesis testing as employed by
\citet{DBLP:journals/corr/abs-1905-02383} (possibly using a combination
of Propositions A.3 and 2.12 as a starting point) could be employed. Our
elementary analysis, particularly the monotonicity Lemma
\ref{lem:monotonicity}, suggests a strategy for numerically computing
scales for given privacy parameters.\footnote{Implemented in
  https://github.com/laats/SubbotinMechanism} Furthermore, it allows for
the simple reduction to the one dimensional case employed in the proof
of the extension to the multidimensional case in Lemma \ref{lem:iso}.

\begin{remark}[Setup]\label{rem:setup}

Throughout this section, we let \(f(x) = e^{-\psi(x)}\) be a probability
density where \(\psi : \ensuremath{\mathbb{R}}\to (-\infty, \infty]\) is
convex and even, which we can take as a definition of log-concavity of
\(f\). We will adopt the convention that \(f\) lower semi-continuous,
i.e., the set \(\{x \in \ensuremath{\mathbb{R}}\mid f(x) > t\}\) for
some threshold \(t\) is open, and the set
\(\{x \in \ensuremath{\mathbb{R}}\mid f(x) \leq t\}\) is closed. Now,
\(\operatorname{Supp}(f) = (-a, a)\) for some \(a \in (0, \infty]\)
where
\(\operatorname{Supp}(f) := \{x \in \ensuremath{\mathbb{R}}\mid f(x) > 0\}\)
denotes the support of \(f\). \qedhere

\end{remark}

We first prove a small technical result we rely on later.

\begin{lemma}\label{lem:supports}

Let \(f\) with support \((-a, a)\) be as in the setup Remark
\ref{rem:setup}, \(\infty > d \geq 0\), and let
\(f_{+}(x) := f(x - d)\). Then \begin{align*}
d \geq 2a &\implies \operatorname{Supp}(f) \cap \operatorname{Supp}(f_{+}) = \emptyset \\
d < 2a &\implies \left[ \frac{d}{2}, a\right) \in \operatorname{Supp}(f) \cap \operatorname{Supp}(f_{+})
\end{align*} both hold.

\end{lemma}

\begin{proof}[\bf\em Proof of Lemma \ref{lem:supports}]

Note that \(\operatorname{Supp}(f) = (-a, a)\) and
\(\operatorname{Supp}(f_{+}) = (-a + d, a + d)\). When \(d \geq 2a\),
then \(-a + d \geq a\), so
\(\operatorname{Supp}(f) \cap \operatorname{Supp}(f_{+}) = \emptyset\).
On the other hand, when \(d < 2a\), then \(d/2 < a\), so
\(\{-d/2, d/2\} \subseteq \operatorname{Supp}(f)\). From
\(f(-d/2) = f(d/2 - d) = f_{+}(d/2)\), we get that
\(d/2 \in \operatorname{Supp}(f) \cap \operatorname{Supp}(f_{+})\). From
the overlap of supports we can conlude
\(\left[ d/2, a\right) \in \operatorname{Supp}(f) \cap \operatorname{Supp}(f_{+})\).
\qedhere

\end{proof}

We now follow the proof outline and show that the set \(A^{c}\) can be
expressed in terms of a threshold.

\begin{lemma}[Thresholding]\label{lem:threshold}

Let \(f\) with support \((-a, a)\) be as in the setup Remark
\ref{rem:setup}, \(\infty > d \geq 0\), \(\epsilon \geq 0\), and let
\(f_{+}(x) := f(x - d)\). Then for \begin{align*}
A^{c} &:= \{ z \in \ensuremath{\mathbb{R}}\mid f_{+}(z) > e^{\epsilon} f(z)\} \\
A^{c}_{-} &:= A^{c} \cap \operatorname{Supp}(f) \\
A^{c}_{+} &:= \operatorname{Supp}(f_{+}) - \operatorname{Supp}(f) \\
\end{align*} each of the following hold.

\begin{enumerate}
\def\labelenumi{\arabic{enumi}.}
\tightlist
\item
  \(A^{c}_{-} \cap A^{c}_{+} = \emptyset\)
\item
  \(A^{c} = A^{c}_{-} \cup A^{c}_{+}\)
\item
  if \(d \geq 2a\) then \(A^{c}_{-} = \emptyset\) and
  \(A^{c}_{+} = \operatorname{Supp}(f_{+})\)
\item
  if \(d < 2a\) then, \begin{align*}
  A^{c}_{-} = \emptyset &\implies A^{c} = \{ z \in \operatorname{Supp}(f_{+}) \mid z
  \geq a \}, \\
  A^{c}_{-} \neq \emptyset &\implies A^{c} = \left\{ z \in \left( \frac{d}{2},a \right) \mid z
  > z^{*} \right\}
  \end{align*} where \[
  z^{*} = \max \left\{  z \in \left( \frac{d}{2},a \right) \mid f(z - d) =
  e^{\epsilon} f(z) \right\}.
  \]
\end{enumerate}

\end{lemma}

\begin{proof}[\bf\em Proof of Lemma \ref{lem:threshold}]

Note

\begin{enumerate}
\def\labelenumi{\alph{enumi}.}
\tightlist
\item
  Necessarily, \(A^{c}_{+} \cap \operatorname{Supp}(f) = \emptyset\).
\item
  \(z \in A^{c} \implies z \in \operatorname{Supp}(f_{+})\) since we
  need \(f_{+}(z) > 0\) for \(f_{+}(z) > e^{\epsilon} f(z)\).
\item
  \(z \in A^{c}_{+} \implies z \in A^{c}\) since here \(f(z) = 0\) and
  \(f_{+}(z) > 0\).
\end{enumerate}

We conclude 1. from a. and
\(A^{c}_{-} \subseteq \operatorname{Supp}(f)\).

We conclude 2. from b. and c.~together with 1.

When \(d \geq 2a\), we have from Lemma \ref{lem:supports} that
\(\operatorname{Supp}(f) \cap \operatorname{Supp}(f_{+}) = \emptyset\).
From b. and c.~we then conclude 3.

When \(d < 2a\), the first implication of 4. follows from 2. Now assume
\(A^{c}_{-} \neq \emptyset\). We can then write
\(A^{c}_{-} = \left\{ z \in \operatorname{Supp}(f) \mid \frac{f_{+}(z)}{f(z)} > e^{\epsilon} \right\}\).
From Lemma \ref{lem:supports} we have that
\(d/2 \in \operatorname{Supp}(f) \cap \operatorname{Supp}(f_{+})\).
Since \(\frac{f_{+}(d/2)}{f(d/2)} = 1 \leq e^{\epsilon}\), we have that
\(d/2 \not\in A^{c}_{-}\). By Lemma \ref{lem:mlr} we then get
\(x \in A^{c}_{-} \implies x > d/2\). Since
\(A^{c}_{-} \neq \emptyset\), there exists a \(z \in (d/2, a)\) such
that \(f_{+}(z) > e^{\epsilon} f(z)\). Therefore, by lower
semi-continuity of \(f\) we have \(\sup A = \max A\) for
\(A = \{ z \in (d/2, a) \mid f_{+}(z) = e^{\epsilon} f(z) \}\). From
this, we conclude the second implication in 4. \qedhere

\end{proof}

\begin{lemma}[Privacy criterion]\label{lem:probscrit}

Let \(f\) with support \((-a, a)\) be as in the setup Remark
\ref{rem:setup}, \(\infty > d \geq 0\), \(\epsilon \geq 0\), and let
\(f_{+}(x) := f(x - d)\). Furthermore, let
\(A^{c} = \{ z \in \ensuremath{\mathbb{R}}\mid f_{+}(z) > e^{\epsilon} f(z) \}\),
and let \(Y_{+} \sim f_{+}\) and \(Y \sim f\) be two random variables.
Then both the following hold.

\begin{enumerate}
\def\labelenumi{\arabic{enumi}.}
\tightlist
\item
  If \(d \geq 2a\), then \[
  P(Y_{+} \in A^{c}) - e^{\epsilon} P(Y \in A^{c}) = P(Y_{+} \in A^{c})
  = 1.
  \]
\item
  If \(d < 2a\), then for \(\delta \geq 0\) \begin{align*}
  P(Y_{+} \in A^{c}) &- e^{\epsilon} P(Y \in A^{c}) \leq \delta \\
  &\iff \\
  P(Y_{+} > t) &- e^{\epsilon} P(Y > t) \leq \delta \\
  \end{align*} where \[
  t := \sup \{ z < a \mid f_{+}(z) \leq e^{\epsilon} f(z)\}.
  \]
\end{enumerate}

\end{lemma}

\begin{proof}[\bf\em Proof of Lemma \ref{lem:probscrit}]

Let \(A^{c}_{-}\) and \(A^{c}_{+}\) be defined as in Lemma
\ref{lem:threshold}.

Point 1. follows directly from points 2. and 3. in Lemma
\ref{lem:threshold}.

From points 1. and 3. in Lemma \ref{lem:threshold} we have that \[
P(Y \in A^{c}) = P(Y \in A^{c}_{-}) + P(Y \in A^{c}_{+}) = P(Y \in A^{c}_{-}).
\] Also, note that
\(A^{c}_{-} = \emptyset \iff \left( \forall z \in (\frac{d}{2}, a) \right) f_{+}(z) \leq e^{\epsilon}f(z)\).

If \(A^{c}_{-} = \emptyset\) then, by the above, \begin{align*}
P(Y_{+} \in A^{c}) &- e^{\epsilon} P(Y \in A^{c}) \leq \delta \\
&\iff \\
P(Y_{+} \in A^{c}) &\leq \delta \\
&\iff \\
P(Y_{+} \geq a) &\leq \delta
\end{align*} where the last equivalence is due to Lemma
\ref{lem:threshold} point 4. Since \(P(Y_{-} \geq a) = 0\) always, we
conclude that if \(A^{c}_{-} = \emptyset\), then \begin{align*}
P(Y_{+} \in A^{c}) &- e^{\epsilon} P(Y \in A^{c}) \leq \delta \\
&\iff \\
P(Y_{+} \geq a) &- e^{\epsilon} P(Y \geq a)\leq \delta.
\end{align*} If \(A^{c}_{-} \neq \emptyset\), then by Lemma
\ref{lem:threshold} point 4. \begin{align*}
P(Y_{+} \in A^{c}) &= P(Y_{+} > z^{*}) \\
P(Y_{-} \in A^{c}) &= P(Y_{-} > z^{*})
\end{align*} for
\(z^{*} = \max \left\{ z \in \left(\frac{d}{2}, a\right) \mid f_{+}(z) \leq e^{\epsilon}f(z) \right\}\).
We note that in both cases above, the thresholds \(a\) and \(z^{*}\),
respectively, can be expressed as
\(\sup \{ z < a \mid f_{+}(z) \leq e^{\epsilon} f(z)\}\).

We conclude the proof by noting that \(P(Y_{+} \geq a) = P(Y_{+} > a)\).
\qedhere

\end{proof}

\begin{lemma}[Monotonicity of privacy criterion]\label{lem:monotonicity}

Let \(f\) with support \((-a, a)\) be as in the setup Remark
\ref{rem:setup}, \(\infty > d > 0\), \(\epsilon \geq 0\), \(s > 0\).
Also, let \(F(x) = \int_{-\infty}^{x}f(y) dy\). Then, the function \[
g(\epsilon, d, s):=F \left( \frac{d - t}{s} \right) - 
e^{\epsilon} F \left( -\frac{t}{s} \right)
\] for \[
t := \sup \left\{  z < as \mid f\left( \frac{z - d}{s} \right)
\leq  e^{\epsilon} f\left( \frac{z}{s} \right)\right\} 
\] is

\begin{enumerate}
\def\labelenumi{\alph{enumi}.}
\tightlist
\item
  increasing in \(d\) and strictly so if \(t < as\)
\item
  decreasing in \(s\) and strictly so if \(t < as\)
\item
  decreasing \(\epsilon\) and strictly so if \(t < as\)
\item
  non-negative, and
\item
  zero if and only if \(t = \infty\).
\end{enumerate}

\end{lemma}

\begin{proof}[\bf\em Proof of Lemma \ref{lem:monotonicity}]

We can partition the parameter region
\(W \subseteq [0, \infty)^{2} \times (0, \infty)\) into two disjoint
regions: \(W_{=}\) when \(t = as\) and \(W_{<}\) when \(t < as\).

First assume \((\epsilon, d, s) \in W_{=}\). If \(a = \infty\), then
\(g = 0\), and a. -- e. hold. Now, assume \(a < \infty\). Then
\(g(\epsilon, d, s) = F \left( \frac{d - as}{s} \right)\) and does not
depend on \(\epsilon\), i.e., \(g\) is increasing (but not strictly) in
\(\epsilon\). Note \(F(z)\) is increasing in \(z\). Therefore, since
\(\frac{d - as}{s} = \frac{d}{s} - a\) is increasing in \(d\) and
decreasing in \(s\), we have that \(g\) is increasing in \(d\) and
decreasing in \(s\). Since \(F\) is non-negative, \(g\) is also
non-negative, and also zero only if \(d = 0\). Consequently, a. -- e.
all hold in this case.

Now assume \((\epsilon, d, s) \in W_{<}\). Let
\(\alpha=\frac{t - d}{s}\) and \(\beta= \frac{t}{s}\). Then viewing
\(t = t(\epsilon, d, s) < as\) and applying the chain rule for
differentiation and rearranging, we get \begin{align*}
\frac{\partial}{\partial d} g =
&-\frac{1}{s}\left( f \left(\alpha \right) - {e}^{\epsilon}
f \left(\beta \right) \right) \left(\frac{\partial}{\partial d}t
\left(\epsilon , d , s\right)\right)\\
&+\frac{f \left(\alpha \right)}{2 s}+\frac{e^{\epsilon} f
\left(\beta \right)}{2 s}\\
\frac{\partial}{\partial s} g =
&\frac{1}{s^{2}} \left( f \left(\alpha \right)-{e}^{\epsilon} f \left(\beta \right)\right) t \left(\epsilon ,
d , s\right)\\
&- \frac{1}{s} \left(f \left(\alpha \right)-{e}^{\epsilon} f \left(\beta \right)\right)
\left(\frac{\partial}{\partial s}t \left(\epsilon , d ,
s\right)\right)\\
&-\frac{f \left(\alpha \right) d}{2 s^{2}}-\frac{{
e}^{\epsilon} f \left(\beta \right) d}{2 s^{2}},\\
\frac{\partial}{\partial \epsilon} g =
&\;e^{\epsilon}(F(\beta) - 1) - \frac{1}{s} \left(\frac{\partial}{\partial
\epsilon} t(\epsilon, d, s)\right) \left( f(\alpha) - e^{\epsilon} f(\beta)\right)
\end{align*} Since \(f\) is lower semicontinous, we have that
\(f(\alpha) - e^{\epsilon}f(\beta) = 0\) when \(t < as\). Furthermore,
since \(t < as \implies t \in (-as, as)\), we get that
\(\beta = t/s \in (-a, a)\) and consequently \(f(\beta) > 0\) and
\(F(\beta) < 1\). Then, since \(d > 0\), \(s > 0\), \(f \geq 0\), and
\(F \leq 1\), \begin{align*}
\frac{\partial}{\partial d} g &=
\frac{f \left(\alpha \right)}{2 s}+\frac{e^{\epsilon} f
\left(\beta \right)}{2 s} > 0\\
\frac{\partial}{\partial s} g &=
- \left( \frac{f \left(\alpha \right) d}{2
s^{2}}+\frac{{e}^{\epsilon} f \left(\beta \right) d}{2 s^{2}} \right) < 0\\
\frac{\partial}{\partial \epsilon} g &= e^{\epsilon}(F(\beta) - 1) < 0.
\end{align*} We have now proven parts a., b., c.~and the ``if'' part of
e. We now turn to part d.~

We first note that we can write \(g\) as
\(F(-\alpha) - e^{\epsilon}F(-\beta)\). In order for \(g\geq 0\), we
must have that \[
\frac{F(-\alpha)}{F(-\beta)} \geq e^{\epsilon}.
\] Now, let \(a(x) = (x - d)/s\), \(b(x) = x/s\), and
\(\tau = t(\epsilon, d, s)\). Then we have that \(\alpha = a(\tau)\) and
\(\beta = b(\tau)\). Using that \(F(-x) = \int_{x}^{\infty} f(w)dw\) and
integration by substitution we get \begin{align*}
F(-\alpha) &= F(-a(\tau)) = s^{-1}\int_{\tau}^{\infty} f(a(x)) dx \\
F(-\beta) &= F(-b(\tau)) = s^{-1}\int_{\tau}^{\infty} f(b(x)) dx.
\end{align*} Recalling that the likelihood ratio
\(r(x) = f(a(x))/f(b(x))\) is non-decreasing and that
\(r(\tau) = e^{\epsilon}\), we have that
\(f(a(x)) \geq e^{\epsilon} f(b(x))\) for \(x \geq \tau\). Consequently,
we can conclude d.~as \[
\frac{F(-\alpha)}{F(-\beta)} \geq e^{\epsilon}.
\] Now assume that \(f(a(x)) = e^{\epsilon} f(b(x))\) for all
\(x \geq \tau\). Then
\(t(\epsilon, d, s) = \sup\{z \mid r(z) \leq e^{\epsilon}\} = \infty\).
Consequently, when \(t(\epsilon, d, s) < \infty\) we must have that
\(\int_{\tau}^{\infty}f(a(x)) > \int_{\tau}^{\infty}e^{\epsilon} f(b(x))\).
This in turn means that \(\frac{F(-\alpha)}{F(-\beta)} > e^{\epsilon}\)
which implies \(g > 0\) when \(t(\epsilon, d, s) < \infty\). This
concludes the proof of the ``only if'' part of e., and the proof of the
lemma. \qedhere

\end{proof}

\begin{proof}[\bf\em Proof of Lemma \ref{lem:neccsuff}]\label{proof:lem:neccsuff}

Let \(x = q(d)\) and \(y = q(d')\), \(d = x - y\), \(Y_{+} = sX + d\),
\(Y_{-} = s X\), \(f_{+}(x) = s^{-1}f(s^{-1}(x - d))\), and
\(f_{-}(x) = s^{-1}f(s^{-1}x))\). Then, variables \(Y_{+}\) and
\(Y_{-}\) are distributed according to densities \(f_{+}\) and
\(f_{-}\), respectively.

By Lemma \ref{lem:translation}, in order to show
\(\Pr(s X + x \in S) \leq e^{\epsilon} \Pr(s X + y \in S) + \delta\) for
all measurable sets \(S\), it suffices to show that
\(\Pr(Y_{+} \in S) \leq e^{\epsilon} \Pr(Y_{-} \in S) + \delta\) for all
measurable sets \(S\). The latter is equivalent to showing that for all
measurable \(S\),
\(\Pr(Y_{+} \in S) - e^{\epsilon} \Pr(Y_{-} \in S) \leq \delta\). Since
\(f\) is even, we can apply Lemma \ref{lem:symmetry} to let \(d \geq 0\)
without loss of generality for the remainder of this proof.

Let \(A^{c} = \{x \mid f_{+}(x) > e^{\epsilon} f_{-}(x) \}\) Then, from
Lemma \ref{lem:equiv} we have that \begin{align}
\label{eq:neccsuffgen}
\Pr(Y_{+} \in A^{c}) - e^{\epsilon} \Pr(Y_{-} \in A^{c}) \leq \delta
\end{align} is a sufficient and necessary condition for
\(\Pr(s X + x \in S) \leq e^{\epsilon} \Pr(s X + y \in S) + \delta\) for
all measurable sets \(S\).

Since \(f\) is even, we have that \(F(-x) = 1 - F(x)\). This allows us
to write \(\Pr(Y_{+} > t) = F((d-t)/s)\) and
\(\Pr(Y_{-} > t) = F(-t/s)\). Furthermore,
\(f((z-d)/s) \leq e^{\epsilon} f(z) \iff \psi(z/s) - \psi((z - d)/s) \leq \epsilon\).
This, together with the chaining of Lemma \ref{lem:probscrit} and Lemma
\ref{lem:monotonicity} conclude the proof. \qedhere

\end{proof}

\begin{proof}[\bf\em Proof of Lemma \ref{lem:standards}]\label{proof:lem:standards}

From Inspecting (\ref{eq:neccsuff}), it suffices to show that for
\(\Delta t(\epsilon, 1, s) = t(\epsilon, \Delta, s \Delta)\) for
\(\Delta>0\). Let
\(t_{z} := \sup \left\{ z \mid \frac{f\left(\frac{z-1}{s}\right)}{f\left( \frac{z}{s} \right)} \leq e^{\epsilon}\right\}\).
Substituting \(u=z\Delta\) into \(t_{z}\), we obtain
\(t_{u} = \sup \left\{u \mid \frac{f\left(\frac{u-\Delta}{s\Delta}\right)}{f\left( \frac{u}{s\Delta} \right)}\leq e^{\epsilon}\right\}\).
From this we see that \(\Delta t_{z} = t_{u}\). That
\(s(\epsilon, \delta, \Delta)\) is non-increasing in both \(\epsilon\)
and \(\delta\) is due to Lemma \ref{lem:monotonicity}. \qedhere

\end{proof}

\hypertarget{proofs-from-section}{%
\subsection{\texorpdfstring{Proofs from Section
\ref{sec:selected}}{Proofs from Section }}\label{proofs-from-section}}

\begin{proof}[\bf\em Proof of Theorem \ref{thm:laplace}.]\label{proof:thm:laplace}

For the standard Laplace distribution,
\(\psi(x) = \psi_{1} = \log(2) + |x|\), which is convex and even. Then
\begin{align*}
\psi\left(\frac{x}{s}\right) - \psi\left(\frac{x-\Delta}{s}\right)  = 
{\left| \frac{x}{s}\right|}-{\left| \frac{\Delta-x}{s}\right|} = \epsilon
\end{align*} has solutions \begin{align*}
S &= \begin{cases}
\{ x \mid x \geq \Delta \}, & \Delta = \epsilon s \\
\left\{\frac{\epsilon  s}{2}+\frac{\Delta}{2}\right\}, & \Delta > \epsilon s\\
\emptyset, & \text{otherwise}.
\end{cases} 
\end{align*} When \(\Delta = \epsilon s\), we have that \(t = \infty\),
and applying Lemma \ref{lem:neccsuff}, we have
\((\epsilon, 0)\)-differential privacy when \(s \geq \Delta/\epsilon\).
Now, let \(\Delta > \epsilon s\). Letting \(F(x) = F_1(x)\) (the
standard Laplace cdf) and substituting
\(\frac{\epsilon s}{2}+\frac{\Delta}{2}\) for \(t\) in
(\ref{eq:neccsuff}) and solving for \(s\) yields \begin{align*}
s \geq \frac{\Delta}{\epsilon - 2 \log\left(1-\delta \right)}.
\end{align*} For \(\delta=0\), the above coincides with
\(s \geq \Delta/\epsilon\). The theorem follows from Lemma
\ref{lem:neccsuff}. \qedhere

\end{proof}

\begin{proof}[\bf\em Proof of Theorem \ref{thm:logistic}.]\label{proof:thm:logistic}

We have that \(f_{\log}(x) = e^{-\psi(x)}\) for
\(\psi(x)= \psi_{\log} = x +2 \log\left(1+e^{-x}\right)\). Since
\(\psi''(x) = \frac{2 e^{-x}}{(1 + e^{-x})^{2}} \geq 0\), we have that
\(\psi\) is convex. Evenness can be seen from
\(e^{\psi(x)} = 2 + e^{-x} + e^{x}\) and the fact that \(e^{x}\) is a
strictly monotone function. Then, the equation \begin{align*}
\psi\left(\frac{x}{s}\right)& -
\psi\left(\frac{x-\Delta}{s}\right) \\&=
\frac{2 \log\left(1+e^{-\frac{x}{s}}\right) s -2
\log\left(1+e^{\frac{-x +\Delta}{s}}\right) s
+\Delta}{s} \\&= \epsilon
\end{align*} solved for \(x\) has solution \begin{align*}
x^{*} &= -\log\left(
A-1
\right) s, \qquad \text{where}\\
A &= e^{-\frac{-2 s \log\left(
e^{\frac{\Delta}{s}}-1\right)+2 s
\log\left(-1+e^{\frac{\epsilon  s +\Delta}{2 s}}\right)-\epsilon  s
+\Delta}{2 s}}.
\end{align*} Let \(t = x^{*}\), and let
\(F(x) = \int_{-\infty}^{x} f_{\log}(w) dw\). Then, manipulation yields
\begin{align*}
F &\left( \frac{\Delta -t}{s} \right) - 
e^{\epsilon}F \left( -\frac{t}{s} \right) \\
&=
\frac{e^{\frac{\epsilon  s +\Delta}{s}}+e^{\frac{2 \Delta}{s}}+2
e^{\frac{\epsilon  s +\Delta}{2 s}}-2 e^{\frac{\epsilon  s +3
\Delta}{2 s}}-e^{\frac{\Delta}{s}}-e^{\epsilon}}{\left({
e}^{\frac{\Delta}{s}}-1\right)^{2}}. 
\end{align*} Solving \[
\frac{e^{\frac{\epsilon  s +\Delta}{s}}+e^{\frac{2 \Delta}{s}}+2
e^{\frac{\epsilon  s +\Delta}{2 s}}-2 e^{\frac{\epsilon  s +3
\Delta}{2 s}}-e^{\frac{\Delta}{s}}-e^{\epsilon}}{\left({
e}^{\frac{\Delta}{s}}-1\right)^{2}} = \delta
\] for \(s\) yields two solutions \[
\frac{\Delta}{2 \log\left(
\frac{e^{\frac{\epsilon}{2}}
\pm
\sqrt{\delta  \left(e^{\epsilon}+\delta -1\right)}}{1-\delta}\right)}
\] of which the larger drops out due to monotonicity of the left hand
side of (\ref{eq:neccsuff}). The smaller produces the stated bound. The
theorem then follows from Lemma \ref{lem:neccsuff}. \qedhere

\end{proof}

\begin{proof}[\bf\em Proof of Theorem \ref{thm:balle}.]\label{proof:thm:balle}

The standard Gaussian density is \(f_{2}(x) = e^{-\psi(x)}\) for convex
and even
\(\psi(x) = \psi_{2}(x) = \frac{\log\left(2 \pi \right)}{2}+\frac{x^{2}}{2}\).
Since
\(\psi \left( \frac{z}{\sigma} \right) - \psi \left( \frac{z-d}{\sigma} \right) = \left(\frac{\Delta }{\sigma^{2}}\right) z -\frac{\Delta^{2}}{2 \sigma^{2}}\),
we see that \(\psi\) is MLR-unbounded. The equation
\(\psi \left( \frac{z}{\sigma} \right) - \psi \left( \frac{z-d}{\sigma} \right) = \epsilon\)
has unique solution
\(z^{*} = \frac{2 \epsilon \sigma^{2}+\Delta^{2}}{2 \Delta}\) for
\(\sigma> 0\) and \(\Delta > 0\). The theorem then follows from Lemma
\ref{lem:neccsuff}. \qedhere

\end{proof}

\hypertarget{proofs-from-section-1}{%
\subsection{\texorpdfstring{Proofs from Section
\ref{sec:utility}}{Proofs from Section }}\label{proofs-from-section-1}}

\begin{lemma}\label{lem:lmono}

Let \(\psi : \ensuremath{\mathbb{R}}\to \ensuremath{\mathbb{R}}\) be
even, convex, and non-constant. For \(d > 0\),
\(l(z, s) = \psi(z/s) - \psi((z - d)/s)\) is increasing in \(z\) and
strictly decreasing in \(s > 0\) for \(z > d/2\).

\end{lemma}

\begin{proof}[\bf\em Proof of Lemma \ref{lem:lmono}.]\label{proof:lem:lmono}

For \(d > 0\), \(l\) is increasing in \(z\) by Lemma \ref{lem:mlr}. For
\(l\) to be decreasing in \(s\), we need that if \(s < s'\) then
\(l(z,s) - l(z,s') \geq 0\). We have \begin{align*}
l(z,s) - l(z,s') &= a - b
\end{align*} for
\(a = \psi \left( \frac{z}{s} \right) - \psi \left( \frac{z-d}{s} \right)\)
and
\(b=\psi \left( \frac{z}{s'}\right) - \psi \left(\frac{z-d}{s'} \right)\).
Now assume that \(z > d/2\), which by non-constancy, convexity and
eveness of \(\psi\) ensures that \(\psi(z/s) > \psi((z-d)/s)\) for any
\(s>0\), so \(a > 0\) and \(b > 0\). We now show that \(a > b\) to round
out the proof. Let \(x_{s1} = (z-d)/s\), \(x_{s2} = z/s\),
\(x_{s'1} = (z-d)/s'\), and \(x_{s'2} = z/s'\). Then,
\(x_{s2} - x_{s1} = d/s > d/s' =  x_{s'2} - x_{s'1}\),
\(x_{s'1} < x_{s1}\), and \(x_{s'2} <  x_{s2}\). Consequently (by
\(z > d/2\), evenness ,convexity, and non-constancy) we have that the
slope of the line passing through \((x_{s'1},  \psi(x_{s'1}))\) and
\((x_{s'2}, \psi(x_{s'2}))\) is positive and at most that of the line
passing through \((x_{s1}, \psi(x_{s1}))\) and
\((x_{s2}, \psi(x_{s2}))\). Then since
\(x_{s2} - x_{s1} = d/s > d/s' =  x_{s'2} - x_{s'1}\), we must have that
\(a = \psi(x_{s2}) - \psi(x_{s1}) > b = \psi(x_{s'2}) - \psi(x_{s'1})\).
\qedhere

\end{proof}

\begin{proof}[\bf\em Proof of Theorem \ref{thm:unboundeds}.]\label{proof:thm:unboundeds}

Let \(l(z, s) = \psi(z/s) - \psi((z - d)/s)\). Since \(\psi\) is even
and convex, \(l\) is increasing in \(z\) from Lemma \ref{lem:mlr}. Since
\(f\) is MLR-unbounded, there will always exist \(z\) such that
\(\psi(z/s) - \psi((z - d)/s) > \epsilon\) for any \(\epsilon > 0\).
This, in turn means that \(t \in \operatorname{Supp}(f)\). From Lemma
\ref{lem:monotonicity} we then have that
\(F((\Delta - t)/s) - e^{\epsilon}F(-t/s)\) is strictly decreasing in
\(s\), non-negative, and 0 if and only if
\(t = \sup \{ z \mid l(z,s) \leq \epsilon \} = \infty\). Noting that we
only need to consider \(z > d/2\) (From the proof of Lemma
\ref{lem:probscrit}), applying Lemma \ref{lem:lmono}, we can make \(t\)
arbitrarily large by increasing \(s\). Consequently,
\(\lim_{s\to\infty}F((\Delta - t)/s) - e^{\epsilon}F(-t/s) = 0\), and we
conclude the Lemma. \qedhere

\end{proof}

\hypertarget{proofs-from-section-2}{%
\subsection{\texorpdfstring{Proofs from Section
\ref{sec:logvslap}}{Proofs from Section }}\label{proofs-from-section-2}}

\begin{proof}[\bf\em Proof of Theorem \ref{thm:s1sl}.]\label{proof:thm:s1sl}

Applying Theorem \ref{thm:laplace} and Theorem \ref{thm:logistic}, we
see that \begin{align}
\label{eq:rat}
\rho(\epsilon, \delta) &= 
\frac{2 \log\left(\sqrt{\delta} \sqrt{{e}^{\epsilon}+\delta -1}+{e}^{\frac{\epsilon}{2}}\right)-2
\log\left(1-\delta \right)}{\epsilon -2
\log\left(1-\delta \right)}.
\end{align} For \(1>\delta\geq 0\) we get that \[
2\log(2) + \epsilon > 2\log\left(\sqrt{\delta} \sqrt{{e}^{\epsilon}+\delta
-1}+{e}^{\frac{\epsilon}{2}}\right) \geq \epsilon.
\] Substituting, the lower bound into (\ref{eq:rat}) we get \[
\rho(\epsilon, \delta) \geq 1.
\] Renaming \(\epsilon\) as \(x\) for readability, we now write
\begin{align*}
\rho(x, \delta) &= \frac{f(x)}{g(x)} \;
\text{where} \\
f(x) &= -2 \log\left(\sqrt{\delta} \sqrt{e^{x}+\delta
-1}+e^{\frac{x}{2}}\right)\\&\phantom{+2}+2 \log\left(1-\delta
\right) \leq 0\\
g(x) &= -x +2 \log\left(1-\delta \right) \leq 0, \text{with} \\
f'(x) &= -\frac{2 \left(\frac{\sqrt{\delta} e^{x}}{2
\sqrt{e^{x}+\delta -1}}+\frac{{
e}^{\frac{x}{2}}}{2}\right)}{\sqrt{\delta} \sqrt{{
e}^{x}+\delta -1}+e^{\frac{x}{2}}} \leq 0\\
g'(x) &= -1 \leq 0,
\end{align*} with all the inequalities strict for \(x > 0\). For
\[H(x) = \frac{g(x)^{2}(f/g)'(x)}{|g'(x)|}\] we now get \begin{align*}
\lim_{x \to 0}H(x) &=h(\delta) \quad \text{for}\\ h(\delta) &= \frac{\left(2 \delta -2\right)
\log\left(1-\delta \right)-2 \log\left(\delta +1\right)
\left(\delta +1\right)}{\delta +1},
\end{align*} and \begin{align*}
h'(\delta) &= \frac{4 \log\left(1-\delta \right)}{\left(\delta
+1\right)^{2}} < 0,
\end{align*} which means that for \(1 > \delta \geq 0\),
\(h(\delta) \leq 0\).

We have that \begin{align*}
(f'/g')'&(x) = \\q(x)& \frac{\sqrt{\delta} e^{\frac{x}{2}}}{2 \left(e^{x}+\delta
-1\right)^{\frac{3}{2}} \left(\sqrt{\delta} \sqrt{{
e}^{x}+\delta -1}+e^{\frac{x}{2}}\right)^{2}}, 
\end{align*} for \begin{align*}
q(x) &= \left(-2
e^{\frac{x}{2}}
\left(\sqrt{\delta}-\delta^{\frac{3}{2}}\right) \sqrt{{
e}^{x}+\delta -1}\right)\\&\phantom{X}+\left(\delta-1 \right) \left(\delta +2 {
e}^{x}-1\right),
\end{align*} and recognize that \((f'/g')'(x) \leq 0\) if
\(q(x) \leq 0\). Since \(\sqrt{\delta}-\delta^{\frac{3}{2}} \geq 0\) for
\(1 > \delta \geq 0\), we also recognize that \(q(x) \geq 0\).

A combination of Theorems 4 and 5 in Anderson et
al.\citep{10.2307/27642062}, yields: If \(g'\) never vanishes on an open
interval \((0, b) \in \ensuremath{\mathbb{R}}\) and \(gg'>0\) on
\((0, b)\), \(\lim_{x\to 0} H(x) \leq 0\), and \(f'/g'\) is decreasing
on \((0, b)\), then \((f/g)' < 0\) on \((0, b)\). By the above, all the
conditions are met and we can conclude that \(\rho(\epsilon, \delta)\)
is decreasing in \(\epsilon\) for \(\epsilon \in (0, b)\) for arbitrary
\(b > 0\).

Now, \[
\lim_{\epsilon \to 0} \rho(\epsilon,\delta) = \rho(0,\delta) =
\frac{-\log\left(\delta +1\right)+\log\left(1-\delta \right)}{\log\left(1-\delta \right)}
\] and \begin{align*}
\frac{d}{d \delta} \rho(0,
\delta) &=
\frac{\log\left(\frac{1-\delta}{\delta +1}\right) \delta +\log\left(\frac{1}{\left(\delta +1\right) \left(1-\delta \right)}\right)}{\log\left(1-\delta \right)^{2} \left(1-\delta \right) \left(\delta +1\right)},
\end{align*} which we recognize as being non-positive, meaning that the
ratio decreases in \(\delta\). Consequently, \begin{align*}
\rho(\epsilon, \delta) < \lim_{\delta\to 0} \rho(0,
\delta) = 2,
\end{align*} which concludes the proof. \qedhere

\end{proof}

\begin{proof}[\bf\em Proof of Corollary \ref{cor:v1vl}]\label{proof:cor:v1vl}

Follows from \(\operatorname{Var}(X_{1}) = 2\),
\(\operatorname{Var}(X_{\log}) = \frac{\pi^{2}}{3}\), \begin{align*}
\frac{\operatorname{Var}(s_{1}(\epsilon, \delta,\Delta)X_{1})}{\operatorname{Var}(s_{\log}(\epsilon,
\delta,\Delta)X_{\log})} &= \left( \rho_{1,\log}(\epsilon, \delta) \right)^{2} \frac{\operatorname{Var}(X_{1})}{\operatorname{Var}(X_{\log})},
\end{align*} and Theorem \ref{thm:s1sl}. \qedhere

\end{proof}

\hypertarget{proofs-from-section-3}{%
\subsection{\texorpdfstring{Proofs from Section
\ref{sec:extensions}}{Proofs from Section }}\label{proofs-from-section-3}}

\begin{proof}[\bf\em Proof of Lemma \ref{lem:iso}]\label{proof:lem:iso}

Let \(d > 0\), and let \begin{align*}
D^{(n)}_{d} &= \{ \ensuremath{\mathbf{x}} \in \ensuremath{\mathbb{R}}^{n} \mid \| \ensuremath{\mathbf{x}} \| = d\}, \\
A^{(n)}_{\ensuremath{\mathbf{d}}} &= \{ \ensuremath{\mathbf{x}} \in \ensuremath{\mathbb{R}}^{n} \mid f(\|\ensuremath{\mathbf{x}} - \ensuremath{\mathbf{d}}
\|) \leq e^{\epsilon} f(\| \ensuremath{\mathbf{x}} \|) \}, \\
A^{(n)}_{d} &= \bigcup_{\ensuremath{\mathbf{d}} \in D^{(n)}_{d}} A^{(n)}_{\ensuremath{\mathbf{d}}}.
\end{align*} We note that for \(\ensuremath{\mathbf{d}} \in D^{1}_{d}\)
we have \begin{align*}
A^{(1)}_{\ensuremath{\mathbf{d}}} &= \{x \in \ensuremath{\mathbb{R}}\mid f(|x-d|) \leq e^{\epsilon}
f(|x|)\}\\ 
&= \{x \in \ensuremath{\mathbb{R}}\mid f(x - d) \leq e^{\epsilon} f(x) \}.
\end{align*} The first equality in the above display is due to isomorphy
and the second is due to \(f\) being even. Therefore, \begin{align*}
A^{(1)}_{d} = \{x \in \ensuremath{\mathbb{R}}\mid f(x - d) \leq e^{\epsilon} f(x) \}.
\end{align*} If we are able to show that
\(\ensuremath{\mathbf{x}} \in A^{(n)}_{d} \iff \|\ensuremath{\mathbf{x}}\| \in A^{(1)}_{d}\),
then the proof follows by the argument in the proof of Lemma
\ref{lem:neccsuff} for \(A^{c} = \left(A^{(1)}_{d}\right)^{c}\).

We start with \((\implies)\). First note
\(\ensuremath{\mathbf{x}} \in A^{(n)}_{d} \implies \forall \ensuremath{\mathbf{d}} \in D^{(n)}_{d} (f(\| \ensuremath{\mathbf{x}}-\ensuremath{\mathbf{d}} \|) \leq e^{\epsilon}f(\| \ensuremath{\mathbf{x}} \|))\),
and in particular for
\(\ensuremath{\mathbf{y}} = \arg\min_{\ensuremath{\mathbf{d}} \in D^{(n)}_{d}} \| \ensuremath{\mathbf{x}} - \ensuremath{\mathbf{d}}\| = c \ensuremath{\mathbf{x}}\)
for \(c = \frac{d}{\| \ensuremath{\mathbf{x}} \|}\). We also have that
\begin{align}
f(\| \ensuremath{\mathbf{x}} -\ensuremath{\mathbf{y}} \|) &\geq \max_{\ensuremath{\mathbf{d}} \in D^{(n)}_{d}} f(\|
\ensuremath{\mathbf{x}} - \ensuremath{\mathbf{d}} \|), \label{eq:yminfmax}
\end{align} since \(f\) is non-increasing on
\(\ensuremath{\mathbb{R}}_{\geq 0}\). Now, \begin{align*}
\| \ensuremath{\mathbf{x}} - \ensuremath{\mathbf{y}} \| &= \|  \ensuremath{\mathbf{x}} - c \ensuremath{\mathbf{x}} \| = \| (1 - c)
\ensuremath{\mathbf{x}} \| \\
& = |1-c| \|\ensuremath{\mathbf{x}}\| \\
&=
\begin{cases}
\| \ensuremath{\mathbf{x}} \| - d & \text{if $d \leq \|\ensuremath{\mathbf{x}}\|$},\\
d - \| \ensuremath{\mathbf{x}} \| & \text{if $d > \|\ensuremath{\mathbf{x}}\|$}.
\end{cases} 
\end{align*} Since \(f\) is even,
\(f(\| \ensuremath{\mathbf{x}} \| - d) = f(d - \| \ensuremath{\mathbf{x}} \|)\),
and consequently \[
f(\| \ensuremath{\mathbf{x}} -\ensuremath{\mathbf{y}} \|) \leq e^{\epsilon} f(\| \ensuremath{\mathbf{x}} \|) \implies
f(\| \ensuremath{\mathbf{x}} \| - d) \leq e^{\epsilon} f(\| \ensuremath{\mathbf{x}} \|).
\] By (\ref{eq:yminfmax}) we conclude the \((\implies)\) case.

We now turn to \((\impliedby)\). By the reverse triangle inequality we
have
\(\|\ensuremath{\mathbf{x}} - \ensuremath{\mathbf{d}}\| \geq \|\ensuremath{\mathbf{x}}\| - \|\ensuremath{\mathbf{d}}\|\)
and since \(f\) is non-increasing on
\(\ensuremath{\mathbb{R}}_{\geq 0}\),
\(f(\|\ensuremath{\mathbf{x}} - \ensuremath{\mathbf{d}}\|) \leq f(\|\ensuremath{\mathbf{x}}\| - \|\ensuremath{\mathbf{d}}\|)\).
From this we conclude the \((\impliedby)\) case. \qedhere

\end{proof}

\begin{proof}[\bf\em Proof of Theorem \ref{thm:gensub}]\label{proof:thm:gensub}

Applying a reparameterization of \(\tau=(2/q)^{1/q}\) in (12) in
\citet{10.2307/2337215}, they show that when the independent
\(X_{i} \sim f_{p}\), where \(f_{p}\) is the
\ensuremath{\text{Subbotin}_{p}} density, the variable
\(\ensuremath{\mathbf{X}}\) has a \(\|\cdot\|_{p}\)-spherical
distribution. Since the distribution of \(\ensuremath{\mathbf{X}}\) is
\(\|\cdot\|_{p}\)-spherical, \(\ensuremath{\mathbf{X}}\) is distributed
according to density
\(g(\ensuremath{\mathbf{x}}) = h(\| \ensuremath{\mathbf{x}} \|_{p})\)
for a function \(h\). Since the marginals \(X_{i} \sim f_{p}\), we must
have that
\(g(\ensuremath{\mathbf{x}}) = C^{-1} f_{p}(\|\ensuremath{\mathbf{x}}\|_{p})\)
for a normalizing constant \(C\). The theorem then follows from Lemma
\ref{lem:iso}. \qedhere

\end{proof}

\hypertarget{proofs-from-section-4}{%
\subsection{\texorpdfstring{Proofs from Section
\ref{sec:subopt}}{Proofs from Section }}\label{proofs-from-section-4}}

\begin{proof}[\bf\em Proof of Theorem \ref{thm:inorm}]

Suppose \(\|f(R)\|_{\infty} = 0\), in which case it is an
\(l_{\infty}\)-ball and \(\Delta_{p} = 0\). Now, suppose
\(\|f(R)\|_{\infty} > 0\) and let
\((v,w) \in \arg\sup_{(x,y) \in f(R)^{2}}\|x-y\|_{\infty}\). Since
\(\|f(R)\|_{\infty} > 0\), there exists distinct \(a,b \in R\) such that
\(f(a) = v\) and \(f(b) = w\). We can now let \(d_{1} \in \{a\}^{n}\)
and \(d_{2} \in \{a\}^{n-1} \times \{b\}\). Then,
\(\|q(d_{1}) - q(d_{2})\|_{p} = |\nu(n)| \left( \sum_{i=1}^{m}|v_{i} - w_{i}|^{p} \right)^{1/p} \leq m^{1/p}|\nu(n)|\|v - w\|_{\infty} = m^{1/p} \,|\nu(n)|\,\|f(R)\|_{\infty}\).
By linearity of \(f\), we have that
\(\|q(d_{1}) - q(d_{2})\|_{p} \geq \|q(d) - q(d')\|_{p}\) for any
neighboring \(d,d' \in R^{n}\). If \(f(R)\) is a \(l_{\infty}\)-ball, we
can choose \(v\) and \(w\) such that
\(|v_{i} - w_{i}| = \|v - w\|_{\infty}\) for all \(i\). Then we get the
claimed equality. \qedhere

\end{proof}

\end{document}